\documentclass[sigconf]{acmart}

\usepackage{amsmath,amssymb,amsfonts}
\usepackage{algorithmic}
\usepackage{graphicx}
\usepackage{textcomp}
\usepackage{xcolor}
\usepackage{xspace}
\usepackage{cases}
\usepackage{hyperref}
\usepackage{pifont}
\usepackage[justification=centering]{subcaption}
\usepackage[labelfont={bf}, labelsep=period]{caption}
\usepackage{wrapfig}
\usepackage{balance}
\usepackage{framed}
\usepackage{formatting/shortcuts}
\usepackage{mathtools}
\usepackage{pgfplots}

\usepackage{enumitem}
\usepackage{multirow}
\usepackage{listings}
\usepackage{pgf}
\usepackage{tikz}
\usetikzlibrary{shapes,arrows}
\usepackage{dot2texi}
\usepackage{array}
\usepackage{pgfplots}
\usetikzlibrary{patterns}

\usepackage{circuitikz}
\usetikzlibrary{calc}
\usetikzlibrary{positioning}

\usepackage{booktabs,wrapfig}

\colorlet{punct}{red!60!black}
\definecolor{background}{HTML}{EEEEEE}
\definecolor{delim}{RGB}{20,105,176}
\colorlet{numb}{magenta!60!black}

\newtheorem{remark}{Remark}

\begin{document}
\title{Quantitative Verification of Neural Networks\\ 
And Its Security Applications}
\begin{abstract}

Neural networks are increasingly employed in safety-critical domains. This has
  prompted interest in verifying or certifying logically encoded properties of
  neural networks. 
Prior work has largely focused on checking existential properties, wherein the
  goal is to check whether there exists any input that violates a given property
  of interest.
However, neural network training is a stochastic process, and many questions
  arising in their analysis require probabilistic and quantitative reasoning,
  i.e., estimating how many inputs satisfy a given property. To this end, our
  paper proposes a novel and principled framework to quantitative verification
  of logical properties specified over neural networks.
Our framework is the first to provide {\em PAC}-style soundness guarantees, in
  that its quantitative estimates are within a controllable and bounded error
  from the \newtext{true} count.
We instantiate our algorithmic framework by building a prototype tool called
  \tool that enables checking rich properties over \newtext{binarized} neural networks. We
  show how emerging security analyses can utilize our framework in \newtext{$3$} concrete
  point applications: \newtext{quantifying} robustness to adversarial inputs,
  efficacy of trojan attacks, and fairness/bias of given neural networks.

\end{abstract}

\author{Teodora Baluta}
\email{teobaluta@comp.nus.edu.sg}
\affiliation{National University of Singapore}

\author{Shiqi Shen}
\email{shiqi04@comp.nus.edu.sg}
\affiliation{National University of Singapore}

\author{Shweta Shinde}
\email{shwetas@eecs.berkeley.edu}
\affiliation{University of California, Berkeley}
\authornote{Part of the research done while working at National University of Singapore.}

\author{Kuldeep S. Meel}
\email{meel@comp.nus.edu.sg}
\affiliation{National University of Singapore}

\author{Prateek Saxena}
\email{prateeks@comp.nus.edu.sg}
\affiliation{National University of Singapore}

\maketitle              %
\section{Introduction}
\label{sec:intro}

Neural networks are witnessing wide-scale adoption, including in domains with
the potential for a long-term impact on human society. Examples of these domains
include criminal sentencing~\cite{compas}, drug
discovery~\cite{wallach2015atomnet,verbist2015using}, self-driving
cars~\cite{bojarski2016end}, aircraft collision avoidance
systems~\cite{julian2016policy}, robots~\cite{bhattacharyya2015certification},
and drones~\cite{giusti2016machine}. While neural \newtext{networks} achieve
human-level accuracy in several challenging tasks such as image
recognition~\cite{krizhevsky2012imagenet,szegedy2015going,he2016deep} and
machine
translation~\cite{lecun2015deep,sutskever2014sequence,bahdanau2014neural},
studies show that these systems may behave erratically in the
wild~\cite{fredrikson2014privacy,fredrikson2015model,papernot2016limitations,
papernot2016transferability,papernot2015distillation,evtimov2017robust,
uesato2018adversarial,athalye2018obfuscated,tramer2017ensemble,shokri2017membership,
biggio2012poisoning,carlini2016hidden,carlini2018secret}.

Consequently, there has been \newtext{a surge} of interest in the design of
methodological approaches to verification and testing of neural
networks. Early efforts focused on {\em qualitative} verification
\newtext{wherein,} given a neural network $N$ and property $P$, one is concerned
with determining whether there exists an input $I$ to $N$ such that
$P$ is violated~\cite{simonyan2013deep,
  sundararajan2017axiomatic,koh2017understanding,datta2016algorithmic,
  pei2017deepxplore,pulina2010abstraction,ehlers2017formal,narodytska2017verifying,
  katz2017reluplex,huang2017safety,dvijotham2018dual}. While such
certifiability techniques provide value, for instance in demonstrating
the existence of adversarial
examples~\cite{goodfellow2014explaining,papernot2016limitations}, it
is worth recalling that the designers of neural network-based systems
often make a statistical claim of their behavior, i.e., a given system
is claimed to satisfy properties of interest with high probability but
not always.  Therefore, many analyses of neural networks require {\em
  quantitative} reasoning, which determines how many inputs satisfy
\prop.

It is natural to encode properties as well as conditions on inputs or
outputs as logical formulae.  We focus on the following formulation of
{\em quantitative verification}: Given a set of neural networks
$\mathcal{N}$ and a property of interest $\prop$ defined over the
union of inputs and outputs of neural networks in $\mathcal{N}$, we
are interested in estimation of how often $\prop$ is satisfied. In
many critical domains, client analyses often require guarantees that
the computed estimates be reasonably close to the ground truth.  We
are not aware of any prior approaches that provide such formal
guarantees, though the need for quantitative verification has recently
been recognized~\cite{webb2018statistical,seshia2018}.

\paragraph{Security Applications.}
Quantitative verification enables many applications in security
analysis (and beyond) for neural networks. We present $3$ point
applications in which the following analysis questions can be
quantitatively answered:
\begin{itemize}

  \item {\em Robustness}: \newtext{How many adversarial samples does a given neural network
  have? Does one neural network have more adversarial inputs compared to
    another one?}

  \item {\em Trojan Attacks}: \newtext{A neural network can be trained
    to classify certain inputs with ``trojan trigger'' patterns to \newtext{the
    desired label}. How well-poisoned is a trojaned model, i.e., how
    many such trojan inputs does the attack successfully work for?}

\item{\em Fairness}: Does a neural network change its predictions significantly
  when certain input features are present (\newtext{e.g.,} when the input record has gender
    attribute set to ``female'' vs. ``male'')?
\end{itemize}

Note that such analysis questions boil down to estimating how often
some property over inputs and outputs is satisfied.  Estimating
counts is fundamentally different from checking whether a satisfiable
input exists. Since neural networks are stochastically trained, the
mere existence of certain satisfiable inputs is not unexpected. The
questions above checks whether their counts are sufficiently large to
draw statistically significant
inferences. Section~\ref{sec:applications} formulates these analysis
questions as logical specifications.

\paragraph{Our Approach.}
The primary contribution of this paper is a new analysis framework, which models
the given set of neural networks $\mathcal{N}$ and $\prop$ as set of logical
constraints, $\varphi$, such that the problem of quantifying how often
$\mathcal{N}$ satisfies $\prop$ reduces to model counting over $\varphi$. We
then show that the quantitative verification is $\#P$-hard. Given the
computational intractability of $\#P$, we  seek to compute rigorous estimates
and introduce the notion of {\em approximate quantitative verification}: given a
prescribed tolerance factor $\varepsilon$ and confidence parameter $\delta$, we
estimate how often $\prop$ is satisfied with \newtext{PAC-style} guarantees, i.e.,
computed result is within a \newtext{multiplicative} $(1+\varepsilon)$ factor of
the ground truth with confidence at least $1-\delta$.

Our approach works by encoding the neural network into a logical
formula in CNF form. The key to achieving soundness guarantees is
our new notion of \newequi, which defines a principled way of encoding
neural networks into a CNF formula $F$, such that quantitative
verification reduces to counting \newtext{the satisfying assignments} of $F$
projected to a subset of the support of $F$. We then use approximate
model counting on $F$, which has seen rapid advancement in
practical tools that provide PAC-style guarantees on counts for $F$.
The end result is a {\em quantitative verification procedure for
  neural networks with soundness and precision guarantees}.

While our framework is more general, we instantiate our analysis
framework with a sub-class of neural networks called binarized neural
networks (or BNNs)~\cite{hubara2016binarized}. BNNs are multi-layered
perceptrons with $\texttt{+/-} 1$ weights and step activation
functions. They have been demonstrated to achieve high accuracy for \newtext{a wide}
variety of applications\newtext{~\cite{rastegari2016xnor,mcdanel2017embedded,kung2018efficient}}.  Due to their
small memory footprint and fast inference time, they have 
been deployed in constrained environments such as embedded
devices~\cite{mcdanel2017embedded,kung2018efficient}.
We observe that specific existing encodings for BNNs 
adhere to our notion of \newequi 
and implement these in a new tool
called \tool\footnote{The name stands for \textbf{N}eural \textbf{P}roperty
\textbf{A}pproximate \textbf{Q}uantifier.  The tool will be released as
open-source post-publication.}.  We provide proofs of key correctness and
composability properties of our general approach, as well as of our specific encodings.  Our
encodings are \newtext{linear} in the size of $\mathcal{N}$ and $\mathcal{P}$.

\paragraph{Empirical Results.}
\newtext{ We show that \tool scales to BNNs with $1-3$ internal layers
  and $50-200$ units per layer. We use $2$ standard datasets namely
  MNIST and UCI Adult Census Income dataset. We encode a total of $84$
  models, each with $6,280-51,410$ parameters, into \totalformulae formulae and quantitatively
  verify them. \tool encodes properties in less
  than a minute and solves $97.1$\% formulae in a $24$-hour timeout.
  Encodings scale linearly in the size of the models, and its running
  time is not dependent on the true counts.
We showcase how \tool can be used in diverse security applications with case
studies. First, we quantify the model robustness by measuring how many
adversarially perturbed inputs are misclassified, and then the effectiveness of
$2$ defenses for model hardening with adversarial training. Next, we evaluate
the effectiveness of trojan attacks outside the chosen test set. Lastly, we
measure the influence of $3$ sensitive features on the output and check if the
model is biased towards a particular value of the sensitive feature.  }

\paragraph{Contributions.}
We make the following contributions:
\begin{itemize}
	\item {\em New Notion.} 
	We introduce the notion of {\em approximate quantitative verification} to
	estimate how often a property $P$ is satisfied by the neural net $N$ with
	theoretically rigorous PAC-style guarantees.
	\item {\em Algorithmic Approach, Tool, \& Security Applications.} 
  We propose a principled algorithmic approach for encoding neural networks to
  CNF formula that preserve model counts. We build an end-to-end tool called
  \tool that can handle BNNs. \newtext{We demonstrate security applications of
  \tool in  quantifying robustness, trojan attacks, and fairness. }
	\item {\em Results.} 
  \newtext{
  We evaluate \tool on \totalformulae
  formulae derived from properties over BNNs trained on two datasets.
We show that \tool presently scales to BNNs of over $50,000$ parameters,
and evaluate its performance characteristics with respect to different
    user-chosen parameters.
  } 
\end{itemize}

\section{Problem Definition}
\label{sec:problem}

\begin{definition}[Specification ($\spec$)]
  Let $\nnset = \{f_1, f_2, \ldots , f_m\}$ be a set of $m$ neural nets, where
  each neural net $f_i$ takes a vector of inputs $\vec{x_i}$ and outputs a
  vector $\vec{y_i}$, such that $\vec{y_i} = f_i(\vec{x_i})$. Let
  $\prop:\{\vec{x} \cup \vec{y}\} \rightarrow \{0, 1\}$ denote the property
  $\prop$ over the inputs  $\vec{x} =  \bigcup\limits_{i=1}^{m}\vec{x_i}$  and
  outputs $\vec{y} =\bigcup\limits_{i=1}^{m}\vec{y_i}$. We define the
  specification of property $\prop$ over $\nnset$ as $\spec(\vec{x}, \vec{y})=
  (\bigwedge\limits_{i=1}^{m} (\vec{y_i} = f_i(\vec{x_i})) \land \prop(\vec{x},
  \vec{y}))$.
\end{definition}

We show several motivating property specifications in
Section~\ref{sec:applications}. For the sake of illustration here,
consider $\nnset= \{f_1, f_2\}$ be a set of two neural networks that
take as input a vector of three integers and output a $0/1$, i.e.,
$f_1 : \mathbb{Z}^3 \rightarrow \{0, 1\}$ and $f_2 : \mathbb{Z}^3
\rightarrow \{0, 1\}$. We want to encode a property to check the
dis-similarity between $f_1$ and $f_2$, i.e., counting for how many
inputs (over all possible inputs) do $f_1$ and $f_2$ produce differing
outputs. The specification is defined over the inputs $\vec{x}=[x_1,
  x_2, x_3]$, outputs $y_1 = f_1(\vec{x})$ and $y_2 = f_2(\vec{x})$ as
$\spec(x_1, x_2, x_3, y_1, y_2) = (f_1(\vec{x}) = y_1 \land
f_2(\vec{x}) = y_2 \land y_1 \neq y_2)$.

Given a specification $\spec$ for a
property $\prop$ over the set of neural nets $\nnset$, a verification
procedure returns $r= 1$ (SAT)  if there exists a satisfying assignment
$\tau$ such that $\tau \models \spec$, otherwise it returns $r = 0$ (UNSAT).
A satisfying assignment for $\spec$ is defined as $\tau : \{\vec{x} \cup
\vec{y}\} \rightarrow \{0, 1\}$ such that $\spec$ evaluates to true, i.e.,
$\spec(\tau)=1$ or $\tau \models \spec$.

While the problem of standard (qualitative) verification asks whether there exists a satisfying
assignment to $\spec$, the problem of quantitative verification asks how many
satisfying assignments or models does $\spec$ admit. We denote the set of
satisfying assignments for the specification $\spec$ as $\setsat{\spec}
=\{\tau:\tau \models \spec\}$.

\begin{definition}[Neural Quantitative Verification ($\nqv$)]
  Given a specification $\spec$ for a property $\prop$ over the set of neural
  nets $\nnset$, a quantitative verification procedure, $\nqv(\spec)$, returns
  the number of satisfying assignments of $\spec$, $r = |\setsat{\spec}|$.
\end{definition}

It is worth noting that $|\setsat{\spec}|$ may be intractably large to
compute via na\"{\i}ve enumeration. For instance, we consider neural
networks with hundreds of bits as inputs for which the unconditioned
input space is $2^{|\vec{x}|}$.  In fact, we prove that quantitative
verification is \#P-hard, as stated
below.

\begin{theorem}
  $\nqv(\spec)$ is \#P-hard, where $\spec$ is a specification for a property
  \prop over binarized neural nets.
\end{theorem}

Our proof is a parsimonious reduction of model counting of CNF
formulas, \#CNF, to quantitative verification of binarized neural
networks. We show how an arbitrary CNF formula $\text{F}$ can be
transformed into a binarized neural net $\bnnfunc$ and a property
\prop such that for a specification $\spec$ for $\prop$ over $\nnset =
\{\bnnfunc\}$, it holds true that $\setsat{\text{F}}=\setsat{\spec}$.
See Appendix~\ref{sec:hardness} for the full proof.

\begin{remark}\label{rem:np-rp}
	The parsimonious reduction from \#CNF to NQV implies that fully 
  polynomial time randomized approximation schemes, including those based 
  on Monte Carlo, cannot exist unless NP=RP.
\end{remark}

The computational intractability of \#P necessitates a search for
relaxations of $\nqv$. To this end, we introduce the notion of an
approximate quantitative verifier that outputs an approximate count
within $\epsilon$ of the true count with a probability greater than $1
- \delta$.

\begin{definition}[Approximate $\nqv$ ($\anqv$)]
  Given a specification $\spec$ for a property $\prop$ over a set of neural nets
  $\nnset$, $0 < \epsilon \leq 1$ and $0 < \delta \leq 1$, an approximate
  quantitative verification procedure $\anqv(\spec, \epsilon, \delta)$ computes
  $r$ such that $Pr[(1+\epsilon)^{-1}|\setsat{\spec}| \leq r \leq
  (1+\epsilon)|\setsat{\spec}|] \geq 1-\delta$.
\end{definition}

The security analyst can set the ``confidence'' parameter $\delta$ and
the precision or \newtext{``error tolerance''} $\epsilon$ as
desired. The $\anqv$ definition specifies the end guarantee of
producing estimates that are \newtext{statistically} sound with respect to chosen
parameters $(\epsilon,\delta)$.

\paragraph{Connection to computing probabilities.} 
Readers can naturally interpret $|\setsat{\spec}|$ as a measure of
probability.  Consider $\mathcal{N}$ to be a set of functions defined
over input random variables $\vec{x}$. The property specification
$\spec$ defines an event that conditions inputs and outputs to certain
values, which the user can specify as desired. The measure
$|\setsat{\spec}|$ counts how often the event occurs under all
possible values of $\vec{x}$. Therefore,
$\frac{|\setsat{\spec}|}{2^{|\vec{x}|}}$ is the probability of the
event defined by $\spec$ occurring. Our formulation presented here
computes $|\setsat{\spec}|$ weighting all possible values of $\vec{x}$
equally, which implicitly assumes a uniform distribution over all
random variables $\vec{x}$. Our framework can be extended to
\newtext{weighted
counting~\cite{ermon2013embed,ermon2013taming,ermon2013optimization,chakraborty2015weighted}}, assigning different user-defined weights
to different values of $\vec{x}$, which is akin to specifying a
desired probability distributions over $\vec{x}$. However, we consider
this extension as a promising future work.

\section{Security Applications}
\label{sec:applications}

We present three concrete application contexts which highlight how
quantitative verification is useful to diverse security analyses.  The
specific property specifications presented here derived directly from
recent works, highlighting that \tool is broadly applicable to analysis
problems actively being investigated.

\paragraph{Robustness.}
An adversarial example for a neural network is an input which under a small
perturbation is classified
differently~\cite{szegedy2013intriguing,goodfellow2014explaining}.
The lower the number of adversarial examples, the more ``robust'' the neural
network.
Early work on verifying robustness aimed at checking whether adversarial inputs
exist. However, recent works suggest that adversarial inputs are statistically
``not
surprising''\newtext{~\cite{uesato2018adversarial,athalye2018obfuscated,ford2019adversarial}}
as they are a consequence of normal error in statistical
classification\newtext{~\cite{gilmer2018adversarial,gilmer2018motivating,mahloujifar2018curse,dohmatob2018limitations}}.
This highlights the importance of analyzing whether a statistically significant
number \newtext{of} adversarial examples exist, not just whether they exist at
all, under desired input distributions.  Our framework allows the analyst to
specify a logical property of adversarial inputs and quantitatively verify it.
Specifically, \newtext{one} can estimate how many \newtext{inputs} are misclassified by the
net ($f$) and within some small perturbation distance $k$ from a benign sample
($\vec{x_b}$) ~\cite{carlini2017towards,papernot2016limitations,
papernot2016transferability}, by encoding the property~\ref{eq:robust_property} in our framework as:
\begin{equation}
  \tag{P1}
  P_1(\vec{x}, \vec{y}, \vec{x_b}, k) = \sum_{j=1}^{|\vec{x}|} (\vec{x_b}[j]
  \oplus \vec{x}[j]) \leq k \land \vec{y_b} \neq \vec{y}
\label{eq:robust_property}
\end{equation}
As a concrete usage scenario, our evaluation reports on BNNs for image
classification (Section~\ref{sec:robustness}). Even for a small given input (say
$m$ bits), the space of all inputs within a perturbation of $k$ bits is ${m
\choose k}$, which is too large to check for misclassification one-by-one. \tool
does not enumerate and yet can estimate adversarial input counts with PAC-style
guarantees \newtext{(Section~\ref{sec:robustness})}. As we permit larger
perturbation, as expected, the number of adversarial samples monotonically
increase, and \tool can quantitatively measure how much. Further, we show how
one can directly compare robustness estimates for two neural networks. Such
estimates may also be used to measure the efficacy of defenses. Our evaluation
on $2$ adversarial training defenses \newtext{shows} that the \newtext{hardened}
models show lesser robustness than the \newtext{plain (unhardened)} model. Such
analysis can help to quantitatively refute, for instance, claims that BNNs are
intrinsically more robust, as suggested in prior
work~\cite{galloway2017attacking}.

\paragraph{Trojan Attacks.}
Neural networks, such as for facial recognition systems, can be
trained in a way that they output a specific value, when the input has
a certain ``trojan trigger'' embedded in
it~\cite{Trojannn,geigel2013neural}. The trojan trigger can be a fixed
input pattern (\newtext{e.g.,} a sub-image) or some transformation that can be
stamped on to a benign image.
One of the primary goals of the trojan attack is to maximize the
number of trojaned inputs which are classified as the desired target
output, $\vec{l_{attack}}$. \tool can quantify the number of such
inputs for a trojaned network, allowing attackers to optimize \newtext{for} this
metric. To do so, one can encode the set of trojaned inputs as all
those inputs $\vec{x}$ which satisfy the following constraint for a
given neural network $f$, trigger $\vec{t}$, $\vec{l_{attack}}$ and
the (pixel) location of the trigger $M$:
\begin{equation}
  \tag{P2}
  P_2(\vec{x}, \vec{y}, \vec{t}, \vec{l_{attack}}, M) = (\vec{x}[j] = \vec{t}[j])
  \land \vec{y} = \vec{l_{attack}},  j \in M \label{eq:trojan_property}
\end{equation}
\newtext{Section~\ref{sec:trojan} shows an evaluation on BNNs trained on
the MNIST dataset. Our evaluation demonstrates that the attack accuracy on
samples from the test set can differ significantly from the total set of
trojaned inputs specified as in property~\ref{eq:trojan_property}.}

\paragraph{Fairness.}
The right notion of algorithmic fairness is being widely
debated\cite{dwork2012fairness,feldman2015certifying,zafar2015fairness,datta2016algorithmic,hardt2016equality,datta2017use}. Our
framework can help quantitatively evaluate desirable metrics measuring
``bias'' for \newtext{neural networks}. Consider a scenario where a neural
network $f$ is used to predict the recommended salary for a new hire
in a company. Having been trained on public data, one may want to
check whether $f$ makes biased predictions based on certain sensitive
features such as race, gender, or marital status of the new hire. To
verify this, one can count how often $f$ proposes a higher salary for
inputs when they have a particular sensitive feature (say ``gender'')
set to certain values (say ``male''), given all other input features
the same.  Formally, this property can be encoded for given sensitive
features $\vec{x_{s_1}} \in \vec{x_1}$, $\vec{x_{s_2}} \in \vec{x_2}$,
where $\vec{x} = \vec{x_1} \cup \vec{x_2}$, along with
values $\vec{s_1}, \vec{s_2}$, as:
\begin{equation}
  \tag{P3}
\begin{aligned}
  P_3(\vec{x}, \vec{y}, \vec{x_{s_1}}, \vec{x_{s_2}}, \vec{s_1},
  \vec{s_2}) = (\vec{x_{s_1}} = \vec{s_1}) \land (\vec{x_{s_2}} =
  \vec{s_2}) \land \\ ((\vec{x_1} - \vec{x_{s_1}}) = (\vec{x_2} -
  \vec{x_{s_2}})) \land \vec{y_1} = \vec{y_2} \label{eq:fair_eq_property}
\end{aligned}
\end{equation}
Notice the \tool counts over {\em all} possible inputs where the
non-sensitive feature remain equal, but only the sensitive feature
changes, which causes no change in prediction. An unbiased model would
produce a very high count, meaning that for most inputs (or with high
probability), changing just the sensitive feature results in no change
in outputs. A follow-up query one may ask is whether setting the sensitive
feature to a certain input value, keeping all other values the same, increases
(or decreases) the output salary prediction.
This can be \newtext{encoded as property~\ref{eq:fair_neq_property}
(or~\ref{eq:fair_neq2_property}) below.}
\begin{equation}
  \tag{P4}
\begin{aligned}
  P_4(\vec{x}, \vec{y}, \vec{x_{s_1}}, \vec{x_{s_2}}, \vec{s_1},
  \vec{s_2}) = (\vec{x_{s_1}} = \vec{s_1}) \land (\vec{x_{s_2}} = 
  \vec{s_2}) \land \\ ((\vec{x_1} - \vec{x_{s_1}}) = (\vec{x_2} -
  \vec{x_{s_2}})) \land \vec{y_1} = HIGH \land \vec{y_2} = LOW \label{eq:fair_neq_property}
\end{aligned}
  \end{equation}
\newtext{
\begin{equation}
  \tag{P5}
\begin{aligned}
  P_5(\vec{x}, \vec{y}, \vec{x_{s_1}}, \vec{x_{s_2}}, \vec{s_1},
  \vec{s_2}) = (\vec{x_{s_1}} = \vec{s_1}) \land (\vec{x_{s_2}} =
  \vec{s_2}) \land \\ ((\vec{x_1} - \vec{x_{s_1}}) = (\vec{x_2} -
  \vec{x_{s_2}})) \land \vec{y_1} = LOW \land \vec{y_2} = HIGH
  \label{eq:fair_neq2_property}
\end{aligned}
\end{equation}
}
\tool can be used to quantitatively verify such properties, and compare models
before deploying them based on such estimates.
\newtext{Section~\ref{sec:fairness} presents more concrete evaluation details
and interpretation of BNNs trained on the UCI Adult dataset~\cite{uci2019}}.

\section{Approach}
\label{sec:approach}

Recall that exact counting (as defined in $\nqv$) is $\#P$-hard.  Even for
approximate counting, many widely used sampling-based approaches, such as based
on Monte Carlo
methods\newtext{~\cite{grosu2005monte,hastings1970monte,neal1993probabilistic,
jerrum1996markov}}, do {\em not} provide soundness guarantees since existence of
a method that only requires polynomially many samples computable in (randomized)
polynomial time would imply $NP=RP$ (See Remark~\ref{rem:np-rp}).  For sound
estimates, it is well-known that many properties encodable in our framework
require intractably large number of samples---for instance, to check for
distributional similarity of two networks $f_1$ and $f_2$ in the classical
model, a lower bound of $O(\sqrt{2^{\vec{x}}})$ samples are needed to obtain
estimates with \newtext{reasonable} $(\epsilon,\delta)$ guarantees.  However,
approximate counting for boolean CNF formulae has recently become practical.
These advances combine the classical ideas of universal hashing with the
advances in the Boolean satisfiability by invoking SAT solvers for NP queries,
i.e., to obtain satisfiable witnesses for queried CNF formulae. The basic idea
behind these approximate CNF counters is to first employ universal hashing to
randomly partition the set of solutions into {\em roughly small} buckets. Then,
the approximate counter can enumerate a tractably small number of witnesses
satisfying $P$ using a SAT solver within one bucket, which calculates the
``density'' of satisfiable solutions in that bucket. By careful analysis using
concentration bounds, these estimates can be extended to the sum over all
buckets, yielding a provably sound PAC-style guarantee of estimates. Our work
leverages this recent advance in approximate CNF counting to solve the problem
of $\anqv$\newtext{~\cite{soos2019bird}}.

\paragraph{The \Newequi framework.} Our key technical advance is a new
algorithmic framework for reducing $\anqv$ to CNF counting with \newtext{an encoding}
procedure that has provable soundness. The procedure encodes $\nnset$ and
$\prop$ into $\spec$, such that model counting in some way over $\spec$ counts
over $\nnset \wedge \prop$. This is {\em not} straight-forward. For
illustration, consider the case of counting over boolean circuits, rather than
neural networks. To avoid exponential blowup in the encoding, often one resorts
to classical {\em equisatisfiable} encoding~\cite{tseitin1983complexity}, which
preserves satisfiability but introduces new variables in the process.
Equisatisfiability means that the original formula is satisfiable if and only if
the encoded one is too. Observe, however, that this notion of equisatisfiability
is {\em not} sufficient for model counting---the encoded formula may be
equisatisfiable but may have many more satisfiable solutions than the original.

We observe that a stronger notion, which we call {\em \newequi},
provides a principled approach to constructing encodings that preserve
counts. An \newequi encoding, at a high level, ensures that the model
count for an original formula can be computed by performing model
counting {\em projected} over the subset of variables in the resulting
formula. We define this \newequi relation rigorously and prove in
Lemma~\ref{lemma:preserve} that model counting over a constraint is
{\em equivalent} to counting over its \newequible encoding. Further,
we prove in Lemma~\ref{lemma:compose} that the \newequi relation is
{\em closed} under logical conjunction. This means model counting over
conjunction of constraints is equivalent to counting over the
conjunction of their \newequible encodings. \Newequi CNF encodings can
thus be composed with boolean conjunction, while preserving \newequi
in the resulting formulae.

With this key observation, our procedure has two remaining sub-steps. First,
we show \newequible encodings for each neural net and properties over them to
individual \newequi CNF formulae. This implies $\specmodel$, the conjunction
of the \newequi CNF encodings of the conjuncts in $\spec$, preserves the
original model count of $\spec$.
Second, we show how an existing approximate model counter for CNF with
$(\epsilon,\delta)$ guarantees  can be utilized to count over a projected subset
of the variables in $\specmodel$.  This end result, by construction, guarantees
that our final estimate of the model count has bounded error, parameterized by
$\varepsilon$, with confidence at least $1-\delta$. 
\begin{figure*}[th]
\center
\begin{minipage}{0.2\textwidth}
\center
\pgfdeclarelayer{background}
\pgfdeclarelayer{foreground}
\pgfsetlayers{background,main,foreground}

\begin{tikzpicture}[scale=0.3]
  \tikzstyle{n} = [draw=blue!50,line width=0.25mm,shape=circle,minimum size=1.5em, inner
  sep=0pt,fill=blue!20]
  \tikzstyle{vars} = [draw=red,line width=0.25mm,shape=circle,minimum size=1.5em, inner
  sep=0pt,fill=red!20]
  \tikzstyle{vars2} = [draw=red,line width=0.25mm,shape=circle,minimum size=1.5em, inner
  sep=0pt,fill=red!40]
  \begin{dot2tex}[dot,tikz,codeonly,styleonly,mathmode,outputdir=paper-dotfiles/]
    \input{dotfiles/bnn_eq1.dot}
  \end{dot2tex}
\end{tikzpicture}

\end{minipage}
\hfill
\begin{minipage}{0.3\textwidth}
\center
\resizebox{1.05\textwidth}{!}{%
  \begin{tabular}{l|ccc|c}
    \multicolumn{1}{l|}{Cardinality Constraints:} & \textbf{$x_1$} & \textbf{$x_2$} &
    \textbf{$x_3$} & \textbf{$f(\vec{x})$} \\ \hline

    \multirow{5}{*}{\begin{tabular}[c]{@{}l@{}}$x_1+x_2+x_3 \geq 2
      \Leftrightarrow v_1 = 1$\\ $x_1+\overline{x}_2+x_3 \geq 1 \Leftrightarrow
      v_2 = 1$\\ $\overline{x}_1+x_2+x_3 \geq 1 \Leftrightarrow v_3 = 1$\\
      $x_1+x_2+\overline{x}_3 \geq 1 \Leftrightarrow v_4 = 1$\\
      $\overline{x}_1+x_2+\overline{x}_3 \geq 1 \Leftrightarrow v_5 = 1$\\ \\
    $v_1+v_2+v_3+v_4+v_5 \geq 5 \Leftrightarrow y$\end{tabular}}
    & 0 & 0 & 0 & 0 \\
     & 0 & 0 & 1 & 0 \\
      & 0 & 1 & 0 & 0 \\
       & 0 & 1 & 1 & 1 \\
        & 1 & 0 & 0 & 0 \\
        \multirow{3}{*}{} & 1 & 0 & 1 & 1 \\
         & 1 & 1 & 0 & 0 \\
          & 1 & 1 & 1 & 1
  \end{tabular}
}
\end{minipage}
\hfill
\begin{minipage}{0.4\textwidth}
\center
\resizebox{0.85\textwidth}{!}{%
  \begin{tabular}{l|ccc|c|c}
    \multicolumn{1}{l|}{Cardinality Constraints:} & \textbf{$x_1$} & \textbf{$x_2$} &
    \textbf{$x_3$} & \textbf{$f_1(\vec{x})$} & $f_2(\vec{x})$\\ \hline
    \multirow{7}{*}{\begin{tabular}[c]{@{}l@{}}
            $\overline{x}_1+\overline{x}_2+x_3 \geq 1 \Leftrightarrow v_5$\\
      $v_1+v_2+v_3+v_4+ v_5 \geq 5 $ \\\hspace*{3cm} $\Leftrightarrow y$ \\
      $f_1: \textcolor{red}{v_2=v_3=v_4=1}$ \\
      $f_2: $\textcolor{red}{$x_1+x_3 \geq 2 \Leftrightarrow v_1$}\\
      $x_1+\overline{x}_2+x_3 \geq 1 \Leftrightarrow v_2$\\
      $\overline{x}_1+x_2+x_3 \geq 1 \Leftrightarrow v_3$\\
      $x_1+x_2+\overline{x}_3 \geq 1 \Leftrightarrow v_4$\\
   \end{tabular}}
    & 0 & 0 & 0 & 0 & 0\\
     & 0 & 0 & 1 & 0 & 0\\
      & 0 & 1 & 0 & 0 & 0\\
       & 0 & 1 & 1 & 1 & \textcolor{red}{0}\\
        & 1 & 0 & 0 & 0 & 0\\
        \multirow{3}{*}{} & 1 & 0 & 1 & 1 & 1\\
         & 1 & 1 & 0 & 0 & 0 \\
          & 1 & 1 & 1 & 1 & 1
  \end{tabular}
}
\end{minipage}
\caption{Example of encoding different BNNs ($f$, $f_1$, $f_2$) as
a conjunction over a set of cardinality constraints.
An attacker manipulates $f$ with the goal to increase the inputs
with trigger $x_3 =1$ that classify as $y=0$.
Specifically, to obtain $f_1$ the weights of $x_1,x_2,x_3$ in constraints of $f$ for $v_2,v_3,v_4$
  are set to 0 (highlighted with dashed lines, on the left).
  To obtain $f_2$, we set $w_{21}=0$.
The trojan property \prop $\doteq (y=0) \land (x_3=1)$ is satisfied by
one input (left) for $f$, whereas for $f_2$ we find two (right).}
  \label{fig:bnn-to-card}
\end{figure*}

\paragraph{Formalization.} We formalize the above notions using notation
standard for boolean
logic\newtext{~\cite{nieuwenhuis2005decision,ganesh2007decision,boigelot2005effective,kozen1983decision}}.
The projection of an assignment $\sigma$ over a subset of the variables
$\vec{t}$, denoted as $\proj{\sigma}{\vec{t}}$, is an assignment of $\vec{t}$ to
the values taken in $\sigma$ (ignoring variables other than $\vec{t}$ in
$\sigma$).

\begin{definition}\label{def:equiwit}
    We say that a formula $\spec:\vec{t}\rightarrow\{0,1\}$ is
    equi-witnessable to a formula $\specmodel:\vec{u} \rightarrow \{0,
    1\}$ where $\vec{t} \subseteq \vec{u}$, if:
    \begin{enumerate}[label=(\alph*)]
        \item
        $\forall \tau \models \spec \Rightarrow$  
        $\exists \sigma, (\sigma \models \specmodel)$ $\land$ 
        $(\proj{\sigma}{\vec{t}} = {\tau})$,
        and
        \item 
        $\forall \sigma \models \specmodel \Rightarrow$
        $\proj{\sigma}{\vec{t}} \models \spec.$
    \end{enumerate}
\end{definition}
An example of a familiar equi-witnessable encoding is Tseitin~\cite{tseitin1983complexity}, which
transforms arbitrary boolean formulas to CNF.
Our next lemma shows that \newequi preserves model counts.  We define
$\setsat{\specmodel}\downarrow\vec{t}$, the set of satisfying assignments of
$\specmodel$ projected over $\vec{t}$, as $\{\proj{\sigma}{\vec{t}}:\sigma
\models \specmodel\}$.

\begin{lemma}[Count Preservation]
\label{lemma:preserve}
    If $\specmodel$ \isnewequi $\spec$, then  
    $|\setsat{\specmodel}\downarrow\vec{t}| = |\setsat{\spec}|$.
\end{lemma}

\begin{proof}
By Definition~\ref{def:equiwit}(a), for every assignment $\tau \models \spec$,
there is a $\sigma \models \specmodel$ and the $\proj{\sigma}{\vec{t}} =
\tau$. Therefore, each distinct satisfying assignment of $\spec$ must have a
unique assignment to $\proj{\sigma}{\vec{t}}$, which must be in
  $\setsat{\specmodel}\downarrow\vec{t}$. It follows that
  $|\setsat{\specmodel}\downarrow\vec{t}| \ge |\setsat{\spec}|$, then.
Next, observe that Definition~\ref{def:equiwit}(b) states that everything in
  $\setsat{\specmodel}\downarrow\vec{t}$ has a satisfying assignment in $\spec$;
that is, its projection cannot correspond to a non-satisfying assignment in
$\spec$.  By pigeonhole principle, it must be that
  $|\setsat{\specmodel}\downarrow\vec{t}| \le |\setsat{\spec}|$. This proves
  that $|\setsat{\specmodel}\downarrow\vec{t}| = |\setsat{\spec}|$.
\end{proof}

\begin{lemma}[CNF-Composibility]\label{lemma:compose}
    Consider $\spec_i:\vec{t_i} \rightarrow \{0,1\}$ and $\specmodel_i:\vec{u_i}
    \rightarrow \{0,1\}$, such that $\spec_i$ \isnewequi $\specmodel_i$, for $i
    \in \{1,2\}$.
    If $\vec{u_1} \cap \vec{u_2} = \vec{t}$, where $\vec{t} = \vec{t_1} \cup \vec{t_2}$, then 
    $\spec_1 \land \spec_2$ \isnewequi $\specmodel_1 \land \specmodel_2$.
\end{lemma}

\begin{proof}
  \begin{enumerate}[label=(\alph*)]
    \item $\forall \tau \models \spec_1 \land \spec_2 \Rightarrow (\tau \models
    \spec_1) \land (\tau \models \spec_2)$. By Definition~\ref{def:equiwit}(a), $\exists
    \sigma_1, \sigma_2, \sigma_1 \models \specmodel_1 \land \sigma_2 \models
    \specmodel_2$. Further, by Definition~\ref{def:equiwit}(a), $\proj{\sigma_1}{\vec{t_1}} =
    \proj{\tau}{\vec{t_1}}$ and $\proj{\sigma_2}{\vec{t_2}} =
    \proj{\tau}{\vec{t_2}}$. This implies that $\proj{\sigma_1}{\vec{t_1} \cap
    \vec{t_2}}$ $=$ $\proj{\sigma_2}{\vec{t_1} \cap \vec{t_2}}$ $=$
    $\proj{\tau}{\vec{t_1} \cap \vec{t_2}}$.
    We can now define the $\sigma_1 \otimes \sigma_2 =
      \proj{\sigma_1}{\vec{u_1} - \vec{t_1}} \cup \proj{\sigma_2}{\vec{u_2}-
      \vec{t_2}}\cup (\proj{\sigma_1}{\vec{t}}\cap\proj{\sigma_2}{\vec{t}})$.
      Since $(\vec{u_1} - \vec{t})$
    $\cap$ $(\vec{u_2} - \vec{t})$ is empty (the only shared variables between
    $\vec{u_1}$ and $\vec{u_2}$ are $\vec{t}$), it follows that $\sigma_1
      \otimes \sigma_2 \models \specmodel_1 \land \specmodel_2$ and that $ \proj{(\sigma_1
    \otimes \sigma_2)}{\vec{t}} = \tau$. This proves part (a) of the claim that
    $\spec_1 \land \spec_2$ \isnewequi $\specmodel_1 \land \specmodel_2$.

    \item $\forall \sigma \models \specmodel_1 \land \specmodel_2$ $\Rightarrow
    (\sigma \models \specmodel_1) \land  (\sigma \models \specmodel_2)$. By
    Definition~\ref{def:equiwit}(b), $\proj{\sigma}{\vec{t_1}} \models \spec_1$ and
    $\proj{\sigma}{\vec{t_2}} \models \spec_2$. This implies
    $\proj{\sigma}{\vec{t}} \models \spec_1 \land \spec_2$, thereby proving the
    part (b) of the definition for the claim that $\spec_1 \land \spec_2$
    \isnewequi $\specmodel_1 \land \specmodel_2$.
  \end{enumerate}
\end{proof}

\paragraph{Final count estimates.}
With the \newtext{CNF-composability} lemma at hand, we decompose the counting
problem over a \newtext{conjunction} of neural networks $\nnset$ and property
$\prop$, to that of counting over the \newtext{conjunction} of their respective
\newequi encodings.  \Newequi encodings preserve counts, and taking their
conjunction preserves counts. It remains to show how to encode $\nnset$ to
boolean CNF formulae, such that the encodings are \newequible. Since the
encoding preserves counts originally desired exactly, we can utilize
off-the-shelf approximate CNF
counters\newtext{~\cite{chakraborty2013scalable,soos2019bird}} which have
$(\epsilon, \delta)$ guarantees. The final counts are thus guaranteed to be
sound estimates by construction, which we establish formally in
Theorem~\ref{thm:main-thm} for the encodings in Section~\ref{sec:design}.

\paragraph{Why not random sampling?}
An alternative to our presented approach is random sampling. One could
simply check what fraction of all possible inputs \newtext{satisfies} $\spec$
by testing on a random set of samples. However, the estimates produced
by this method will satisfy soundness (defined in
Section~\ref{sec:problem}) {\em only if} the events being measured
have sufficiently high probability.
In particular, obtaining such soundness guarantees for rare events, i.e., where
counts may be very low, requires an intractably large number of samples. Note
that such events do arise in security
applications~\cite{carlini2018secret,webb2018statistical}.  Specialized Monte
Carlo samplers for such low probability events have been investigated in such
contexts~\cite{webb2018statistical}, but they do not provide soundness
guarantees. We aim for a general framework, that works irrespective of the
probability of events measured.

\section{\tool Design}
\label{sec:design}

\newtext{Our tool takes as input a set of trained binarized neural networks $\nnset$ and
a property $\prop$ and outputs ``how much" $\prop$ holds over $\nnset$ with
$(\epsilon, \delta)$ guarantees.}
We show a two-step construction from binarized neural nets to CNF. The main
principle we adhere to is that at every step we formally prove that we obtain
\newequible formulas. While BNNs and, in general, neural nets can be encoded
using different background theories, we choose a specialized encoding of BNNs to
CNF. First, we express a BNN using cardinality constraints similar
to~\cite{narodytska2017verifying} (Section~\ref{sec:encoding}). For the second
step, we choose to encode the cardinality constraints to CNF using a
sorting-based encoding (Section~\ref{sec:tocnf}). We prove that \tool is
preserving the \newequi in Theorem~\ref{thm:main-thm}. Finally, we use an
approximate model counter that can handle model counting directly over a
projected subset of variables for a CNF formula~\cite{soos2019bird}.

\subsection{BNN to Cardinality Constraints}
\label{sec:encoding}
\begin{table}[t]
\caption{BNN definition as a set of layers of transformations.}
\resizebox{0.4\textwidth}{!}{%
\begin{tabular}{l}
\toprule
\texttt{\textbf{A. Internal Block} $\bnnblk{k}(\vec{v}_k)=\vec{v}_{k+1}$} \\ 
\toprule
\texttt{1)~Linear Layer} \\
\begin{minipage}{0.45\textwidth}
\begin{align}
t^{lin}_{i} = \langle \vec{v}_k, \vec{w}_{i} \rangle + b_{i}
\label{eq:blk-linear}
\end{align}
\end{minipage}
\\
where $i=1,...,n_{k+1}$, $\vec{w}_{i}$ is the $i_{th}$ column in ${\rm W}_k \in$ \\ $\{-1,1\}^{n_k \times n_{k+1}}$, 
$\vec{b}$ is the bias row vector $\in \mathbb{R}^{n_{k+1}}$ and \\ $\vec{y} \in \mathbb{R}^{n_{k+1}}$ \\ \hline
\texttt{2)~Batch Normalization} \\
\begin{minipage}{0.45\textwidth}
\begin{align}
t^{bn}_{i} = \frac{t^{lin}_{i} - \mu_{k_i}}{\sigma_{k_i}} \cdot \alpha_{k_i}+ \gamma_{k_i}
\label{eq:blk-bn}
\end{align}
\end{minipage}
\\
where $i = 1, ..., n_{k+1}$, $\vec{\alpha}_{k}$ is the $k_{th}$ weight row vector $\in$ \\ 
$\mathbb{R}^{n_{k+1}}$, $\vec{\gamma}_k$ is the bias $\in \mathbb{R}^{n_{k+1}}$, $\vec{\mu}_k \in \mathbb{R}^{n_{k+1}}$ is the mean and \\
$\sigma_k \in \mathbb{R}^{n_{k+1}}$ is the standard deviation. \\ \hline
\texttt{3)~Binarization} \\
\begin{minipage}{0.45\textwidth}
\begin{align}
t^{bn}_{i} \geq 0 \Rightarrow v_{{k+1}_i} = 1 \label{eq:blk-bin-1}\\
  t^{bn}_{i} < 0 \Rightarrow v_{{k+1}_i} = -1 \label{eq:blk-bin-2}
\end{align}
\end{minipage}
\\
where $i = 1, ..., n_{k+1}$. \\ 

\toprule
\texttt{\textbf{B. Output Block} $\bnnout(\vec{v}_d)=\vec{y}$} \\ 
\toprule
\texttt{1)~Linear Layer} \\
\begin{minipage}{0.45\textwidth}
\begin{align}
q^{lin}_i = \langle\vec{v}_d, \vec{w}_{j} \rangle + b_i
\end{align}
\end{minipage}
\\
where $\vec{v}_d \in \{-1, 1\}^{n_d}$, $\vec{w}_{j}$ is the $j_{th}$ column $\in \mathbb{R}^{n_d \times s}$, $\vec{b} \in$ \\$\mathbb{R}^{s}$ is the
bias vector. \\ \hline

\texttt{2)~Argmax} \\
\begin{minipage}{0.45\textwidth}
\begin{align}
y_i = 1 \Leftrightarrow i = \argmax(\vec{q}^{lin})
\end{align}
\end{minipage} 
\\ 
\bottomrule

\end{tabular}
}
\label{table:bnn-def}
\end{table}

Consider a standard BNN $\bnnfunc: \{-1,1\}^n \rightarrow \{0,1\}^s$ that
consists of $d-1$ internal blocks and an output
block~\cite{hubara2016binarized}. We denote the $k$th internal block as
$\bnnblk{k}$ and the output block as $\bnnout$. More formally, given an input
$\vec{x} \in \{-1, 1\}^n$, the binarized neural network is: $\bnnfunc(\vec{x})
= \bnnout(\bnnblk{d-1}(\ldots(\bnnblk{1}(\vec{x})\ldots))$.
For every block $\bnnblk{k}$, we define the inputs to $\bnnblk{k}$ as the vector
$\vec{v_k}$. We denote the output for $k$ block as the vector $\vec{v_{k+1}}$.
For the output block, we use $\vec{v_d}$ to denote its input.
The input to $\bnnblk{1}$ is $\vec{v_1}=\vec{x}$. We summarize the
transformations for each block in Table~\ref{table:bnn-def}.

\begin{table*}[t]
  \caption{\newtext{Encoding for a binarized neural network BNN($\vec{x}$) to cardinality
  constraints, where $\vec{v_1} = \vec{x}$. MILP stands for Mixed Integer Linear
  Programming, ILP stands for Integer Linear Programming.}}
\resizebox{0.82\textwidth}{!}{%
  \begin{tabular}{m{15.5cm}}
\toprule
\texttt{A. $\bnnblk{k}(\vec{v}_k, \vec{v}_{k+1})$ to $\blkf{k}(\vec{v}_k^{(b)},
    \vec{v}_{k+1}^{(b)})$} \\ 
\toprule
\begin{align*}
  \texttt{\milpblk:}~
\frac{
Eq(1), Eq(2), Eq(3), \alpha_{k_i} > 0
}{
\begin{array}[b]{c}
\langle \vec{v}_k, \vec{w}_{i} \rangle \geq -\frac{\sigma_{k_i}}{\alpha_{k_i}}\cdot \gamma_{k_i} + \mu_{k_i} - b_{i}, 
i = 1, ...,n_{k+1}
\end{array}
}~~
  \texttt{\ilpblk:}~
\frac{
\alpha_{k_i} > 0
}{
\begin{array}[b]{c}
\langle \vec{v}_k, \vec{w}_{i} \rangle \geq C_i \Leftrightarrow v_{{k+1}_i} = 1, i = 1,...,n_{k+1} \nonumber \\
\langle  \vec{v}_k, \vec{w}_{i} \rangle < C_i \Leftrightarrow v_{{k+1}_i} =-1,  i = 1,...,n_{k+1} \\
C_i = \lceil -\frac{\sigma_{k_i}}{\alpha_{k_i}}\cdot \gamma_{k_i} +   \mu_{k_i} - b_{i} \rceil 
\end{array}
} \\
  \texttt{\cardblk:}~
\frac{
v^{(b)} = 2v-1, v \in \{-1, 1\}
}{
\begin{array}[b]{c}
\blkf{k}(\vec{v}_k^{(b)},\vec{v}_{k+1}^{(b)}) =
\sum_{j \in w_{k_{i}}^{+}} v_{k_j}^{(b)} + \sum_{j \in w_{k_{i}}^{-}}
{{\overline{v}_{k_j}}^{(b)}}
\geq C_i^{\prime} + |w_{k_{i}}^{-}| \Leftrightarrow v_{{k+1}_i}^{(b)} = 1, 
C_i^{\prime} = \lceil (C_i + \sum_{j=1}^{n_k} w_{ji}) / 2 \rceil
\end{array}
} 
\end{align*}
\\ 
\toprule
\texttt{B. $\bnnout(\vec{v}_d, \vec{y})$ to $\outf(\vec{v}_d^{(b)},\vec{ord}, \vec{y})$} \\ 
\toprule
\begin{align*}
\texttt{Order:}~
\frac{
ord_{ij} \in \{0, 1\}
}{
q^{lin}_i \geq q^{lin}_j \Leftrightarrow ord_{ij} = 1
}~~
  \texttt{\milpout:}~
\frac{
Eq(5),~Eq(Order)
}{
\langle\vec{v}_d, \vec{w}_{i} - \vec{w}_{j} \rangle \geq b_j - b_i \Leftrightarrow ord_{ij} = 1
}~~
  \texttt{\ilpout:}~
\langle\vec{v}_d, \vec{w}_{i} - \vec{w}_{j} \rangle \geq \lceil b_j - b_i \rceil \Leftrightarrow ord_{ij} = 1 \\
  \texttt{\cardout:}~
\frac{
v^{(b)} = 2v-1, v \in \{-1, 1\}
}{
\begin{array}[b]{c}
\outf(\vec{v}_d^{(b)},\vec{ord}, \vec{y}) = 
\Big(
  \big(\sum_{p \in w_{i}^{+} \cap  w_{j}^{-}} v_{d_p}^{(b)} - \sum_{p \in w_{i}^{-}
    \cap w_{j}^{+}} v_{d_p}^{(b)} \geq \lceil E_{ij} / 2 \rceil\big)
    \Leftrightarrow ord_{ij} \nonumber 
    \land \sum_{i=1}^s ord_{ij} = s \Leftrightarrow y_i =
    1\Big), \\
 E_{ij} = \lceil (b_j - b_i + \sum_{p=1}^{n_d} w_{ip} - \sum_{p=1}^{n_d} w_{jp})
  / 2\rceil
\end{array}
}
\end{align*}
\\ 
\toprule
\texttt{C. $\bnnfunc$ to $\cardf$} \\ 
\toprule
\begin{minipage}{\textwidth}
\begin{align}
  \cardf(\vec{x}^{(b)}, \vec{y}, \vec{v}_2^{(b)}, \ldots,
  \vec{v}_d^{(b)},\vec{ord}) = \blkf{1}(\vec{x}^{(b)}, \vec{v_2}^{(b)})
  \bigwedge_{k=2}^{d-1} \Big(\blkf{k}(\vec{v}_k^{(b)}, \vec{v}_{k+1}^{(b)})
  \Big) \nonumber \land \outf(\vec{v}_d^{(b)}, \vec{y}, \vec{ord})
\end{align}
\end{minipage}
\\ 
\bottomrule

\end{tabular}%
}
\label{table:enc-bnn-card}
\end{table*}

\paragraph{Running Example.}
Consider a binarized neural net $f
:\{-1,1\}^3 \rightarrow \{0,1\}$ with a single internal block and a single
output (Figure~\ref{fig:bnn-to-card}).
To show how one can derive the constraints from the BNN's parameters, we work
through the procedure to derive the constraint for $v_1$ or the output of the
internal block's first neuron.
Suppose we have the following parameters: the weight column vector
$\vec{w}_1=[1~1~1]$ and bias $b_1=-2.0$
for the linear layer; $\alpha_1=0.8, \sigma_1=1.0, \gamma_1=2.0$, $\mu_1 =
-0.37$ parameters for the batch normalization layer.
First, we apply the linear layer transformation (Eq.~\ref{eq:blk-linear} in
Table~\ref{table:bnn-def}). We create a temporary variable for this intermediate
output, $t_1^{lin} = \langle \vec{x}, \vec{w}_1 \rangle + b_1 = x_1 + x_2 + x_3
- 2.0$.
Second, we apply the batch normalization (Eq.~\ref{eq:blk-bn} in
Table~\ref{table:bnn-def}) and obtain $t_1^{bn} =
(x_1+x_2+x_3 - 2.0 + 0.37)\cdot 0.8 + 2.0~$.
After the binarization (Eq.~\ref{eq:blk-bin-1} in
Table~\ref{table:bnn-def}), we obtain the constraints $\text{S}_1 = ( (x_1 + x_2
+ x_3 - 2.0 + 0.37)\cdot 0.8 + 2.0 \geq 0)$ and $\text{S}_1
\Leftrightarrow v_1 = 1$.
Next, we move all the constants to the right side of the inequality:
$x_1 + x_2 + x_3 \geq -2.0/0.8 + 2.0 - 0.37\Leftrightarrow v_1 = 1$.
Lastly, we translate the input from the $\{-1, 1\}$ domain to the boolean
domain, $x_i = 2x_i^{(b)} - 1, i \in \{1,2,3\}$, resulting in the following
constraint: $2(x_1^{(b)} + x_2^{(b)} + x_3^{(b)}) - 3 \geq -0.87$.
We use a sound approximation for the constant on the right side to get rid of
the real values and obtain $x_1^{(b)} + x_2^{(b)} + x_3^{(b)} \geq \lceil 1.065
\rceil = 2$.
For notational simplicity the variables $x_1, x_2, x_3$ in
Figure~\ref{fig:bnn-to-card} are boolean variables
(since $x = 1 \Leftrightarrow x^{(b)} = 1)$.

To place this in the context of the security application in
Section~\ref{sec:applications}, we examine the effect of two \newtext{arbitrary trojan attack}
procedures. Their aim is to manipulate the output of a given neural network,
$f$, to a target class for inputs with a particular trigger. Let us
consider the trigger to be $x_3=1$ and the target class $y=0$ for two trojaned
neural nets, $f_1$ and $f_2$ (shown in Figure~\ref{fig:bnn-to-card}).
Initially, $f$ outputs class $0$ for only one input that has the trigger $x_3=1$.
The first observation is that $f_1$ is  equivalent to $f$, even though
its parameters have changed.
The second observation is that $f_2$ changes its output prediction for the input
$x_1=0, x_2=1,x_3=1$ to the target class $0$.
We want \tool to find how much do $f_1$ and $f_2$ change
their predictions for the target class with respect to the inputs that have the
trigger, i.e., $|\setsat{\spec_1}| < |\setsat{\spec_2}|$, where $\spec_1$,
$\spec_2$ are trojan property specifications (property \newtext{$P_2$} as outlined
Section~\ref{sec:applications}).

\paragraph{Encoding Details.}
The details of our encoding in Table~\ref{table:enc-bnn-card} are similar
to~\cite{narodytska2017verifying}. We first encode each block to mixed integer
linear programming and implication constraints, applying the \milpblk rule for
the internal block and \milpout for the outer block
(Table~\ref{table:enc-bnn-card}).
To get rid of the reals, we use sound approximations to bring the constraints
down to integer linear programming constraints (see \ilpblk and \ilpout in Table~\ref{table:enc-bnn-card}).
For the last step, we define a $1:1$ mapping between variables in the
binary domain $x \in \{-1, 1\}$ and variables in the boolean domain $x^{(b)}
\in \{0, 1\}$, $x^{(b)} = 2x - 1$. Equivalently, for $x \in \{-1, 1\}$ there
exists a unique $x^{(b)}$: $  (x^{(b)} \Leftrightarrow ~x = 1) ~\land~
(\overline{x}^{(b)} \Leftrightarrow ~ x = -1)$.
Thus, for every block $\bnnblk{k}(\vec{v}_k)=\vec{v}_{k+1}$, we obtain a
corresponding formula over booleans denoted as $\blkf{k}(\vec{v}_k^{(b)},
\vec{v}_{k+1}^{(b)})$, as shown in rule \cardblk (Table~\ref{table:enc-bnn-card}). 
Similarly, for the output block $\bnnout$ we obtain $\outf(\vec{v}_d,
\vec{ord}, \vec{y})$.
We obtain the representation of $\vec{y} =
\bnnfunc(\vec{x})$ as a formula $\cardf$ shown in Table~\ref{table:enc-bnn-card}. For notational simplicity, we denote the
introduced intermediate variables $\vec{v}_k^{(b)}=[v_{k_1}^{(b)}, \ldots,
v_{k_{n_k}}^{(b)}], k = 2, \ldots, d$ and $\vec{ord}=[ord_i, \ldots, ord_{n_d
\cdot n_d}]$ as $\auxcard$. Since there is a 1:1 mapping between $\vec{x}$ and
$\vec{x}^{(b)}$ we use the notation $\vec{x}$, when it is clear from context
which domain $\vec{x}$ has.
We refer to $\cardf$ as the formula $\cardf(\vec{x},\vec{y},\auxcard)$.

\begin{lemma}\label{lemma:card}
  Given a binarized neural net $\bnnfunc:\{-1,1\}^n \rightarrow \{0,1\}^s$ over
  inputs $\vec{x}$ and outputs $\vec{y}$, and a property $\prop$, let $\spec$ be
  the specification for $\prop$,
  $\spec(\vec{x},\vec{y})=(\vec{y}=\bnnfunc(\vec{x})) \land \prop(\vec{x},
  \vec{y})$, where we represent $\vec{y}=\bnnfunc(\vec{x})$ as
  $\cardf(\vec{x},\vec{y},\auxcard)$. Then $\spec$ \isnewequi
  $\cardf(\vec{x},\vec{y},\auxcard)$.
\end{lemma}

\begin{proof}
  We observe that the intermediate variables for each block in the neural
  network, namely $\vec{v}_k$ for the $k$th block, are introduced by double
  implication constraints. Hence, not only are both part (a) and part (b) of
  definition~\ref{def:equiwit} true, but the satisfying assignments for the
  intermediate variables $\auxcard$ are uniquely determined by $\vec{x}$. Due to
  space constraints, we give our full proof in
  Appendix\newtext{~\ref{appendix:proof-bnn-to-card}}.
\end{proof}

\subsection{Cardinality Constraints to CNF}
\label{sec:tocnf}

Observe that we can \newtext{express} each block in $\cardf$ as a conjunction of
cardinality
constraints\newtext{~\cite{sinz2005towards,asin2011cardinality,abio2013parametric}.
Cardinality constraints are constraints over boolean variables $x_1, \ldots,x_n$
of the form $x_1 +\ldots + x_n \triangle c$ , where $\triangle \in \{=, \leq,
\geq\}$.} More specifically, by applying the \texttt{\cardblk}~rule
(Table~\ref{table:enc-bnn-card}), we obtain a conjunction over cardinality
constraints $\text{S}_{k_i}$, together with an implication:
$\blkf{k}(\vec{v}_k^{(b)},\vec{v}_{k+1}^{(b)}) = \bigwedge_{i=1}^{n_{k+1}}
\cardct_{k_i}(\vec{v}_k^{(b)}) \Leftrightarrow v_{{k+1}_i}^{(b)} $.
We obtain a similar
conjunction of cardinality constraints for
the output block (\cardout, Table~\ref{table:enc-bnn-card}).
The last step for obtaining a Boolean formula representation for the BNN is encoding
the cardinality constraints to CNF.%

We choose cardinality networks~\cite{asin2011cardinality,abio2013parametric}
to encode the cardinality constraints to CNF formulas and show for this
particular encoding that the resulting CNF is \newequible to the cardinality
constraint. Cardinality networks implement several types of gates, i.e., merge
circuits, sorting circuits and 2-comparators, that compose to implement a
merge sort algorithm. More specifically, a cardinality constraint of the form
$\cardct(\vec{x})=x_1 +\ldots + x_n \geq c$ has a corresponding cardinality
network, $\cardcnf{c} = \Big((\text{Sort}_{c}(x_1, \ldots, x_n) = (y_1,
\ldots, y_{c})) \land y_{c}\Big)$, where $\text{Sort}$ is a sorting circuit.
As shown by~\cite{asin2011cardinality,abio2013parametric}, the following holds
true:

\begin{proposition}\label{prop:sort-net}
  A Sort$_{c}$ network with an input of $n$ variables, outputs the first $c$
  sorted bits. $\text{Sort}_{c}(x_1, \ldots, x_n) = (y_1,\ldots, y_{c})$ where
  $y_1 \geq y_2 \geq \ldots \geq y_{c}$.
\end{proposition}

We view $\cardcnf{c}$ as a circuit where we introduce additional variables
to represent the output of each gate, \newtext{and the output of $\cardcnf{c}$ is $1$ only if
the formula $\cardct$ is true. This is similar to how a Tseitin}
transformation~\cite{tseitin1983complexity} encodes a propositional formula
into CNF.

\paragraph{Running Example.}
Revisiting our example in Section~\ref{sec:encoding}, consider $f_2$'s
cardinality constraint corresponding to $v_1$, denoted as
$\text{S}'_1=x_1 + x_3 \geq 2$. This constraint translates to the most basic
gate of cardinality networks, namely a
2-comparator~\cite{batcher1968sorting,asin2011cardinality} shown in
Figure~\ref{fig:card-to-cnf}. Observe that while this efficient encoding
ensures that $S_1$ is equi-satisfiable to the formula $\text{2-Comp}\land y_2$,
counting over the CNF formula does not preserve the count, i.e., it over-counts
due to variable $y_1$. Observe, however, that this encoding is equi-witnessable 
and thus, a projected model count on $\{x_1, x_3\}$ gives the
correct model count of $1$. The remaining constraints shown in
Figure~\ref{fig:bnn-to-card} are encoded similarly and not shown here for brevity.

\begin{figure}
\begin{minipage}{0.1\textwidth}
\center
\resizebox{1.2\textwidth}{!}{%
\begin{tikzpicture}
  \tikzstyle{packet} = [draw=white,line width=0.25mm,shape=circle,minimum
  size=1.5em, inner sep=0pt,fill=white]
  \tikzset{mixing/.style={rectangle, draw, very thick, minimum
  width=12ex,minimum height =7ex,
  rounded corners=1mm, fill=red!70!orange!30,rotate=90}}

  \node[packet] (A) at (0,-1.0) {$x_1$};
  \node[packet] (B) [right=2.5cm of A] {$y_1$};
  \node[mixing] (pro) at ($(A) !.5! (B) + (0.0,-0.5)$) 
    (mult) {2-Comp};

    \node[packet] (C) at (0.0,-2.1) {$x_3$};
    \node[packet] (D) [right=2.5cm of C] {$y_2$};

    \draw[->] (A.east) -- (mult.north|-A.east);
    \draw[->] (mult.south|-A.east) -- (B.west);
    \draw[->] (C.east) -- (mult.north|-C.east);
    \draw[->] (mult.south|-C.east) -- (D.west);
\end{tikzpicture}
}
\end{minipage}
\begin{minipage}{0.35\textwidth}
\center
\resizebox{0.9\textwidth}{!}{%
  \begin{tabular}{l|c|c|c|c|c}
    \multicolumn{1}{c|}{2-Comp Clauses} & $x_1$ & $x_3$ & $y_1$ & $y_2$ & $\text{2-Comp}
    \land y_2$ \\ \hline
    \multirow{4}{*}{\begin{tabular}[c]{@{}l@{}}$\overline{x}_1 \Rightarrow
      \overline{y}_2$\\ $\overline{x}_3 \Rightarrow \overline{y}_2$\\
      $\overline{x}_1 \land \overline{x}_3 \Rightarrow
    \overline{y}_1$\end{tabular}} & 0 & 0 & 0 & 0 & 0 \\
     & 0 & 1 & 0/1 & 0 & 0 \\
      & 1 & 0 & 0/1 & 0 & 0 \\
       & 1 & 1 & 0/1 & 0/1 & 1
  \end{tabular}%
}
\end{minipage}
  \caption{Cardinality networks encoding for $x_1 + x_3 \geq 2$.
  For this case, cardinality networks amount to a 2-comparator gate. Observe
  there are two satisfying assignments for $\text{2-Comp}\land y_2$ due to the
  ``don't care" assignment to $y_1$.}
  \label{fig:card-to-cnf}
\end{figure}
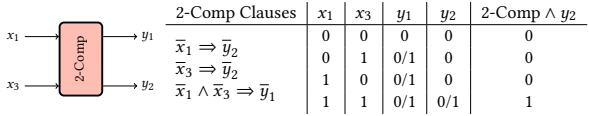

\begin{lemma}[Substitution]\label{lemma:subst}
  Let F be a Boolean formula defined over the variables $\text{Vars}$ and $p \in
  \text{Vars}$. For all satisfying assignments $\tau \models F \Rightarrow
  \proj{\tau}{\text{Vars}-\{p\}} \models F[p\mapsto\tau[p]]$.
\end{lemma}

\begin{lemma}\label{lemma:cnf}
  For a given cardinality constraint, $\cardct(\vec{x}) = x_1 + \ldots + x_n \geq c$,
  let $\cardcnf{c}$ be the CNF formula obtained using cardinality networks,
  $\cardcnf{c}(\vec{x}, \auxcnf) := (\text{Sort}_{c}(x_1, \ldots, x_n) =
  (y_1, \ldots, y_{c}) \land y_{c})$, where $\auxcnf$ are the
  auxiliary variables introduced by the encoding.
  Then, $\cardcnf{c}$ \isnewequi $\cardct$.
  \begin{enumerate}[label=(\alph*)]
    \item $\forall \tau \models \cardct \Rightarrow \exists \sigma, \sigma
      \models \cardcnf{c} \land \proj{\sigma}{\vec{x}} = \tau$.
    \item $\forall \sigma \models \cardcnf{c} \Rightarrow \tau_3
      \vert_{\vec{x}} \models \cardct$.
  \end{enumerate}
\end{lemma}

\begin{proof}
  \begin{enumerate}[label=(\alph*)]
    \item Let $\tau \models \cardct \Rightarrow$ there are least $c$ $x_i$'s
      such that $\tau[x_i] = 1, i \geq c$.
      Thus, under the valuation $\tau_1$ to the input variables $x_1, \ldots,
      x_n$, the sorting network outputs a sequence $y_1, \ldots, y_c$ where
      $y_{c}=1$, where $y_1 \geq \ldots \geq y_c$ (Proposition~\ref{prop:sort-net}).
      Therefore, $\cardcnf{c}[\vec{x} \mapsto \tau]$ =
      $(\text{Sort}_c(x_1 \mapsto \tau[x_1], \ldots, x_n \mapsto
      \tau[x_n]) = (y_1, \ldots, y_{c}) \land y_{c})$ is satisfiable. This
      implies that $\exists \sigma, \sigma \models \cardcnf{c} \land
      \proj{\sigma}{\vec{x}} = \tau$.
    \item Let $\sigma \models \cardcnf{c} \Rightarrow \sigma[y_{c}]=1$.
      By Lemma~\ref{lemma:subst}, $\proj{\sigma}{\vec{x}} \models
      \cardcnf{c}[y_i \\ \mapsto \sigma[y_i]], \forall y_i \in \auxcnf$.
      From Proposition~\ref{prop:sort-net}, under the valuation $\sigma$, there are
      at least $c$ $x_i$'s such that $\sigma[x_i] = 1, i \geq c$. Therefore,
      $\proj{\sigma}{\vec{x}} \models \text{S}$.
  \end{enumerate}

  \end{proof}

For every $\cardct_{k_i}$, $k=1,\ldots,d, i=1,\ldots,n_{k+1}$, we have a CNF
formula $\bnncnf_{k_i}$. The final CNF formula for $\cardf(\vec{x}, \vec{y},
\auxcard)$ is denoted as $\bnncnf(\vec{x}, \vec{y}, \auxall)$, where $\auxall
= \auxcard \bigcup_{k=1}^{d}\bigcup_{i=1}^{n_{k+1}}\auxcnf^{k_i}$ and
$\auxcnf^{k_i}$ is the set of variables introduced by encoding
$\cardct_{k_i}$.

\paragraph{Encoding Size.}
The total CNF formula size is linear in the size of the model.
Given one cardinality constraint $\cardct(\vec{v_k})$, where $|\vec{v_k}|=n$, a cardinality network encoding
produces a CNF formula with $O(n~log^2~c)$ clauses and variables.
\newtext{The constant $c$ is the maximum value that the parameters of the BNN
can take, hence the encoding is linear in $n$.}
\newtext{For a given layer with $m$ neurons, this translates to
$m$ cardinality constraints, each over $n$ variables.
Hence, our encoding procedure produces $O(m \times n)$ clauses and variables for
each layer.
For the output block, $s$ is the number of output classes and $n$
is the number of neurons in the previous layer. Thus, there are $O(s \times s
\times n)$ clauses and variables for the output block.
Therefore, the total size for a BNN with $l$ layers of the CNF is $O(m \times n \times l + s \times s
\times n)$, which is linear in the size of the original model.}

\paragraph{Alternative encodings.}
Besides cardinality networks, there are many other encodings from cardinality
constraints to CNF~\cite{asin2011cardinality,abio2013parametric,abio2011bdds,
sinz2005towards,een2006translating} that can be used as long as they are
\newequible.
We do not formally prove here but we \newtext{strongly suspect} that adder
networks~\cite{een2006translating} and
BDDs~\cite{abio2011bdds} have this property. Adder
networks~\cite{een2006translating} provide a compact, linear
transformation resulting in a CNF with $O(n)$ variables and clauses. The idea
is to use adders for numbers represented in binary to compute the number of
activated inputs and a comparator to compare it to the constant $c$. A 
BDD-based~\cite{een2006translating} encoding builds a BDD representation of the
constraint. It uses $O(n^2)$ clauses and variables.

\subsection{Projected Model Counting}
\label{sec:model-counting}

We instantiate the property $\prop$ encoded in CNF and the neural network
encoded in a CNF formulae $\bnncnf$. We make the powerful observation that we can
directly count the number of satisfying assignment for $\spec$ \newtext{over a subset of
variables, known as projected model counting~\cite{aziz2015exists}}.
\tool uses an approximate model counter with strong PAC-style guarantees.
ApproxMC3~\cite{soos2019bird} is an approximate model counter that can
directly count on a projected formula making a logarithmic number of calls in
the number of formula variables to an NP-oracle, namely a SAT solver.

\begin{theorem}\label{thm:main-thm}
	\tool is an $\anqv$.
\end{theorem}

\begin{proof}
	First, by Lemma~\ref{lemma:compose}, since each cardinality constraint
	$\cardct_{k_i}$ \isnewequi $\bnncnf_{k_i}$ (Lemma~\ref{lemma:cnf}), the
	conjunction over the cardinality constraints is also \newequible. Second, by
	Lemma~\ref{lemma:card}, $\cardf$ \isnewequi $\bnncnf$. Since we use an
	approximate model counter with $(\epsilon, \delta)$
	guarantees~\cite{soos2019bird}, \tool returns $r$ for a given {\rm BNN} and a
	specification $\spec$ with $(\epsilon, \delta)$ guarantees.
\end{proof}

\section{Implementation \& Evaluation}
\label{sec:eval}

We aim to answer the following research questions:
\\
\textbf{(RQ1)} To what extent does \tool scale to, e.g., how large are the neural
nets \newtext{and the formulae that \tool can handle?}
\\
\textbf{(RQ2)} How effective is \tool at providing sound estimates for practical
security applications?
\\
\textbf{(RQ3)} Which factors influence the performance of \tool~\newtext{on our
benchmarks} and how much?
\\
\textbf{(RQ4)} Can \tool be used to refute claims about
security-relevant properties over BNNs?

\paragraph{Implementation.}
We implemented \tool in $4,600$ LOC of Python and C++.
We use the PyTorch (v1.0.1.post2)~\cite{paszke2017automatic}
deep learning platform to train and test binarized neural networks. For
encoding the BNNs to CNF, we build our own tool using %
the PBLib library~\cite{pblib.sat2015} for encoding the cardinality constraints
to CNF. The resulting CNF formula is annotated with a projection set and \tool invokes
the approximate model counter ApproxMC3~\cite{soos2019bird} to count the number
of solutions. We configure a tolerable error $\epsilon=0.8$ and
confidence parameter $\delta=0.2$ as defaults throughout the evaluation. 

\paragraph{Models.} 
Our benchmarks consist of BNNs, on which we tested the properties
derived from the $3$ applications outlined in
Section~\ref{sec:applications}.  The utility of \tool in these
security applications is discussed in Sections~\ref{sec:robustness}-~\ref{sec:fairness}.
For each application, we trained BNNs with the following $4$ different
architectures:
\begin{itemize}
  \item \textbf{ARCH$_1$} - BLK$_1$($100$)
  \item \textbf{ARCH$_2$} - BLK$_1$($50$), BLK$_2$($20$)
  \item \textbf{ARCH$_3$} - BLK$_1$($100$), BLK$_2$($50$)
  \item \textbf{ARCH$_4$} - BLK$_1$($200$), BLK$_2$($100$), BLK$_3$($100$)
\end{itemize}
For each architecture, we take snapshots of the model learnt at
different epochs. In total, this results in $84$ total models,
each with $6,280-51,410$ parameters.
\newtext{\newtext{Encoding} various properties
(Sections~\ref{sec:robustness}-~\ref{sec:fairness}) results in a total of
\totalformulae distinct formulae. For each formula, \tool returns $r$ i.e., 
the number of satisfying solutions.  Given $r$, we calculate \ps i.e., the
percentage of the satisfying solutions with respect to the total input space
size. The meaning of \ps percentage values is application-specific. In trojan attacks,
\pstr represents inputs labeled as the target class.
In robustness quantification, \psadv reports the adversarial
samples.}

\paragraph{Datasets.}
We train models over $2$ standard datasets. Specifically, we quantify robustness and
trojan attack effectiveness on the MNIST~\cite{lecun2010mnist} dataset and
estimate fairness queries on the UCI Adult dataset~\cite{uci2019}.  We choose
them as prior work use these  datasets~\cite{galloway2017attacking, datta2016algorithmic, raghunathan2018certified, gao2019strip,  albarghouthi2017fairsquare}.

\paragraph{\em MNIST.} The dataset contains $60,000$ gray-scale $28 \times 28$ images of
handwritten digits with $10$ classes. In our evaluation, we resize the images to
$10 \times 10$ and binarize the normalized pixels in the images.

\paragraph{\em UCI Adult Census Income.} The dataset is $48,842$ records
with $14$ attributes such as age, gender, education, marital status, occupation,
working hours, and native country. The task is to predict whether a
given individual has an income of over $\$50,000$ a year.  $5 / 14$
attributes are numerical variables, while the remaining attributes are
categorical variables. To obtain binary features, we divide the values of each
numerical variables into groups based on its deviation. Then, we encode each
feature with the least amount of bits that are sufficient to represent each
category in the feature. For example, we encode the race feature which has $5$
categories in total with $3$ bits, leading to $3$ redundant values in this
feature. We remove the redundant values by encoding the property to disable the
usage of these values in \tool. We consider $66$ binary features in
total.

\paragraph{Experimental Setup.}
All experiments are performed on $2.5$ GHz CPUs, $56$ cores, $64$GB RAM. Each
counting process executed on one core and $4$GB memory cap and a $24$-hour
timeout per formula.

\subsection{\tool Benchmarking} 
\label{sec:benchmarks}

We benchmark \tool and report breakdown on \totalformulae formulae.

\paragraph{Estimation Efficiency.}
\newtext{\tool successfully solves $97.1\%$ (\newtext{$1,025$} / \totalformulae) formulae.}
In quantifying the effectiveness of trojan attacks and fairness applications, the
raw size of the input space (over all possible choices of the free variables)
is $2^{96}$ and $2^{66}$, respectively. Naive enumeration for such large
spaces is intractable. \tool returns estimates for $83.3\%$ of the formulae
within $12$ hours and $94.8\%$ of the formulae within $24$ hours for these two
applications. \newtext{In robustness  application, the total input sizes are a
maximum of about $7.5 \times 10^7$.}

\begin{framed}
\vspace{-15pt}
\textbf{Result 1:} \tool solves $97.1\%$ formulae in $24$-hour timeout.
\vspace{-15pt}
\end{framed}

\pgfplotsset{every tick label/.append style={font=\huge}}

\begin{figure}[t]
\resizebox{0.34\textwidth}{!}{%
\begin{tikzpicture}
\begin{axis}[
    ylabel={\newtext{\# of solved formulae}},
    xlabel={Time (hour)},
    label style={font=\huge},
    ymin=350, ymax=1100,
    xmin=0, xmax=24,
    ytick={350, 500, 650, 800, 950, 1100},
    xtick={4, 8, 12, 16, 20, 24},
    legend pos=north west,
    ymajorgrids=true,
    grid style=dashed,
    font=\huge
]
 
\addplot[
    color=RYB2,
    line width=2pt
    ]
    coordinates {
(0.5, 402)
(1.0, 633)
(1.5, 703)
(2.0, 724)
(2.5, 764)
(3.0, 781)
(3.5, 822)
(4.0, 861)
(4.5, 879)
(5.0, 894)
(5.5, 905)
(6.0, 911)
(6.5, 926)
(7.0, 932)
(7.5, 940)
(8.0, 949)
(8.5, 956)
(9.0, 967)
(9.5, 970)
(10.0, 973)
(10.5, 977)
(11.0, 979)
(11.5, 984)
(12.0, 984)
(12.5, 989)
(13.0, 989)
(13.5, 993)
(14.0, 995)
(14.5, 995)
(15.0, 996)
(15.5, 997)
(16.0, 998)
(16.5, 1002)
(17.0, 1007)
(17.5, 1007)
(18.0, 1008)
(18.5, 1011)
(19.0, 1011)
(19.5, 1013)
(20.0, 1014)
(20.5, 1015)
(21.0, 1016)
(21.5, 1018)
(22.0, 1020)
(22.5, 1022)
(23.0, 1024)
(23.5, 1024)
(24.0, 1025)
    };
 
\addplot[dashed, RYB1, line width=2pt] coordinates {(0, 1056) (24, 1056)};
\filldraw[black, font=\huge] (1,670) node[anchor=west] {\totalformulae formulae};
\end{axis}
\end{tikzpicture}
}
  \caption{Number of formulae \tool solves with respect to the time. The solid line represents the aggregate number of formulae \tool solves before the given time. The dashed line represents the total number of formulae.}
  \label{fig:benchmark_time}
\end{figure}
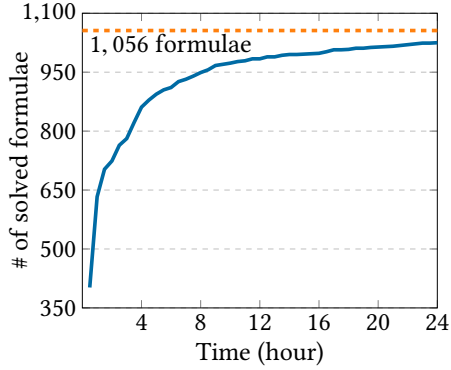

\begin{table*}[hbt!]
\centering
\caption{
Influence of $(\epsilon, \delta)$ on \tool's Performance. The count and time taken 
to compute the bias in ARCH$_2$ trained on UCI Adult dataset  
for changes in values features (marital status, gender, 
and race) i.e., the percentage of individuals whose predicted income
changes from $\leq 50$K to $>50$K when all the other features are same.
\newtext{NLC represents the natural logarithm of the count
\tool generates. Time represents the number of hours \tool
takes to solve the formulae. x represents a timeout.
}
}
\label{tab:epsilon_delta}
\resizebox{0.90\textwidth}{!}{%
\begin{tabular}{|c|c|c|c|c|c|c|c|c|c|c|c|c|c|c|c|c|}
\hline
\multirow{3}{*}{\textbf{Feature}} & \multicolumn{8}{c|}{\textbf{$\delta=0.2$}} & \multicolumn{8}{c|}{\textbf{$\epsilon=0.1$}} \\ \cline{2-17} 
 & \multicolumn{2}{c|}{\textbf{$\epsilon=0.1$}} & \multicolumn{2}{c|}{\textbf{$\epsilon=0.3$}} & \multicolumn{2}{c|}{\textbf{$\epsilon=0.5$}} & \multicolumn{2}{c|}{\textbf{$\epsilon=0.8$}} & \multicolumn{2}{c|}{\textbf{$\delta=0.01$}} & \multicolumn{2}{c|}{\textbf{$\delta=0.05$}} & \multicolumn{2}{c|}{\textbf{$\delta=0.1$}} & \multicolumn{2}{c|}{\textbf{$\delta=0.2$}} \\ \cline{2-17} 
 & \textbf{NLC} & \textbf{Time} & \textbf{NLC} & \textbf{Time} & \textbf{NLC} & \textbf{Time} & \textbf{NLC} & \textbf{Time} & \textbf{NLC} & \textbf{Time} & \textbf{NLC} & \textbf{Time} & \textbf{NLC} & \textbf{Time} & \textbf{NLC} & \textbf{Time} \\ \hline
\textbf{Marital Status} & 
39.10 & 8.79 & 39.08 & 1.35 & 39.09 & 0.80 & 39.13 & 0.34 & x & x & 39.07 & 22.48 & 39.07 & 15.74 & 39.10 & 8.79  \\ \hline
\textbf{Race} & 
40.68 & 3.10 & 40.64 & 0.68 & 40.65 & 0.42 & 40.73 & 0.27 & 40.68 & 14.68 & 40.67 & 8.21 & 40.67 & 5.80 & 40.68 & 3.10  \\ \hline
\textbf{Gender} & 
41.82 & 3.23 & 41.81 & 0.62 & 41.88 & 0.40 & 41.91 & 0.27 & 41.81 & 15.48 & 41.81 & 8.22 & 41.81 & 6.02 & 41.82 & 3.23  \\ \hline
\end{tabular}
}
\end{table*}

\pgfplotsset{every tick label/.append style={font=\huge}}

\begin{figure}[t]
\resizebox{0.34\textwidth}{!}{%
\begin{tikzpicture}
\begin{axis}[
ybar stacked,
symbolic x coords={0-6, 6-12, 12-18, 18-24},
bar width=20pt,
legend style={at={(0.5,-0.20)}, anchor=north,legend columns=-1},
ymin=0, ymax=100,
ylabel={\% of formulae},
xlabel={Time (hour)},
label style={font=\huge},
xtick=data
]
\addplot [fill=RYB1] coordinates
{
(0-6, 55.3)
(6-12, 46.6)
(12-18, 66.7)
(18-24, 47.1)
};
\addplot [fill=RYB2] coordinates
{
(0-6, 34.9)
(6-12, 43.8)
(12-18, 25.0)
(18-24, 47.1)
};
\addplot [fill=RYB3] coordinates
{
(0-6, 9.8)
(6-12, 9.6)
(12-18, 8.3)
(18-24, 5.8)
};
\legend{\strut{\huge  \ps$\leq 10\%$}, \strut{\huge  $10\% <$\ps$\leq 50\%$}, \strut{\huge  \ps$> 50\%$}}
\end{axis}
\end{tikzpicture}
}
\caption{\ps with respect to the time taken by \tool. The size of each region represents the fraction of formulae \tool solves within the specific time for each range of \ps.}
\label{fig:density_time}
\end{figure}
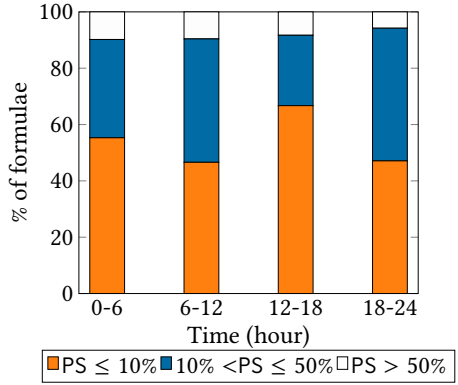

\paragraph{Encoding Efficiency.}
\tool takes a maximum of $1$ minute to encode \newtext{each model}, which is less than
$0.05\%$ of the total timeout. The formulae size  scale linearly with the
model, as expected from  encoding construction. \tool presently utilizes
off-the-shelf CNF counters, and their performance heavily dominates \tool
time. \tool presently scales to formulae of \textasciitilde $3.5\times10^6$
variables and \textasciitilde $6.2\times10^6$ clauses. However, given the
encoding efficiency, we expect   \tool to scale to larger models with future
CNF counters~\cite{chakraborty2016algorithmic, soos2019bird}.

\begin{framed}
\vspace{-5pt}
\textbf{Result 2:} \codename takes \textasciitilde$1$ minute to encode the model.
\vspace{-5pt}
\end{framed}

\paragraph{Number of Formulae vs. Time.}
Figure~\ref{fig:benchmark_time} plots the number of formulae solved with
respect to the time, the relationship  is logarithmic. \tool solves $93.2\%$
formulae in the first $12$ hours, whereas, it only solves $3.9\%$ more in the next
$12$ hours.  We notice that the neural net depth  impacts the performance,
most timeouts ($27/31$) stem from ARCH$_4$.  $26/31$ timeouts are for
Property~\ref{eq:robust_property} (Section~\ref{sec:applications}) to 
quantify adversarial robustness.  Investigating why certain formulae are
harder to count is an active area of independent research~\cite{dudek2017hard,
achlioptas2006solution, dudek2016combine}.

\paragraph{\newtext{Performance with varying ($\epsilon, \delta$).}}
We investigate the relationship between different error and confidence
parameters and test co-relation with parameters that users can pick. We
select a subset of formulae~\footnote{\newtext{Our timeout is $24$ hours per
formula, so we resorted to checking a subset of formulae.}}  which have
varying degrees of the number of solutions, a large enough input space which
is intractable for enumeration, and varying time performance for the
baseline parameters of $\epsilon=0.8,\delta=0.2$. \newtext{Since these arise
most naturally in fairness application encodings using ARCH$_2$, we chose
all the 3 formulae in it.}

We first vary the error tolerance (or precision), $\epsilon \in
\{0.1,0.3,0.5,$ $0.8\}$ while keeping the same $\delta =
0.2$ for the fairness application, as shown in Table~\ref{tab:epsilon_delta}. 
This table illustrates no significant resulting difference in counts reported
by \tool under different precision parameter values. More precisely, the
largest difference as the natural logarithmic of the count is $0.1$ for
$\epsilon=0.3$ and $\epsilon=0.8$ for the feature ``Gender''.  This suggests
that for these formulae, decreasing the error bound does not yield a much
higher count precision.

\newtext{Higher precision does} come at a \newtext{higher} performance cost, as
the $\epsilon=0.1$ takes $16\times$ more time than $\epsilon=0.8$.  The results
are similar when varying the confidence parameter $\delta \in
\{0.2,0.1,0.05,0.01\}$ (smaller is better) for $\epsilon=0.1$
(Table~\ref{tab:epsilon_delta}). This is because the number of calls to the SAT
solver depends only on the $\delta$ parameter, while $\epsilon$ dictates how
constrained the space of all inputs or how small the ``bucket'' of solutions
is~\cite{soos2019bird, chakraborty2013scalable}.  Both of these significantly
increase the time taken. Users can tune $\epsilon$ and $\delta$ based on the
required applications precision  and the available time budget.

\begin{framed}
\vspace{-5pt}
\textbf{Result 3:} \tool reports no significant difference in the counts
 produced when configured with different $\epsilon$ and $\delta$.
\vspace{-5pt}
\end{framed}

\paragraph{\ps vs. Time.}
\newtext{We investigate if \tool solving efficiency varies with increasing
count size.  Specifically, we measure the \ps with respect to the time taken for
all the \totalformulae formulae. Table~\ref{fig:density_time} shows the \ps
plot for $4$ time intervals and $3$ density intervals.  We observe that 
the number of satisfying solutions do not significantly influence  the time
taken to solve the instance. This suggests that \tool is generic enough to
solve formulae with arbitrary solution set sizes.}

\begin{framed}
\vspace{-5pt}
\textbf{Result 4:} For a given $\epsilon$ and $\delta$, \tool solving time is not
significantly influenced by the \ps.
\vspace{-5pt}
\end{framed}

\subsection{Case Study 1: Quantifying Robustness}
\label{sec:robustness}

We quantify the model robustness  and the effectiveness of defenses for model
hardening with adversarial training.

\paragraph{Number of Adversarial Inputs.}
One can count precisely what fraction of inputs, when drawn uniformly at
random from a constrained input space, are misclassified for a given model.
For demonstrating this, we first train $4$ BNNs on the MNIST dataset, one
using each of the architectures ARCH$_1$-ARCH$_4$. We encode the
Property~\ref{eq:robust_property} (Section~\ref{sec:applications})
corresponding to perturbation bound $k \in \{2, 3, 4, 5\}$. We take $30$
randomly sampled images from the test set, and for each one, we encoded one
property constraining adversarial perturbation to each possible value of $k$.
This results in a total of $480$ formulae on which \tool runs with a timeout
of $24$ hours per formula. If \tool terminates within the timeout limit, it
either quantifies the number of solutions or outputs UNSAT, meaning that there
are no adversarial samples with up to $k$ bit perturbation.
\newtext{Table~\ref{tab:robustness} shows the average number of adversarial
samples and their \psadv, i.e., percentage of count to the total input
space.}
\begin{table}[]
\centering
\caption{\newtext{
Quantifying robustness for ARCH$_{1..4}$ and perturbation size from $2$ to $5$.
ACC$_{b}$ represents the percentage of benign samples in the test set
labeled as the correct class. \#(Adv)  and \psadv represent the average
number and percentage of adversarial samples separately. \#(timeout)
represents the number of times \tool timeouts.
}}
\label{tab:robustness}
\resizebox{0.45\textwidth}{!}{%
\begin{tabular}{|c|c|c|c|c|c|}
\hline
  \textbf{Arch} & \textbf{ACC$_{b}$} &
\textbf{\begin{tabular}[c]{@{}c@{}}Perturb\\ Size\end{tabular}} & 
\textbf{\#(Adv)} & \textbf{\psadv} & \textbf{\#(timeout)} \\ \hline
  \multirow{4}{*}{\textbf{ARCH$_1$}} & \multirow{4}{*}{76}
 & $k \leq 2$ & 561 & 11.10 & 0 \\ \cline{3-6} 
 & & $k=3$ & 26,631 & 16.47 & 0  \\ \cline{3-6} 
 & & $k=4$ & 685,415 & 17.48 & 0 \\ \cline{3-6} 
 & & $k=5$ & 16,765,457 & 22.27 & 0 \\ \hline
  \multirow{4}{*}{\textbf{ARCH$_2$}}  &\multirow{4}{*}{79}
 & $k \leq 2$ & 789 & 15.63 & 0 \\ \cline{3-6} 
 & & $k=3$ &  35,156 & 21.74 & 0  \\ \cline{3-6} 
 & & $k=4$ &  928,964 & 23.69 & 0  \\ \cline{3-6} 
 & & $k=5$ & 21,011,934 & 27.91 & 0 \\ \hline
  \multirow{4}{*}{\textbf{ARCH$_3$}}  & \multirow{4}{*}{80}
 & $k \leq 2$ &  518 & 10.25 & 0  \\ \cline{3-6} 
 & & $k=3$ & 24,015 & 14.85 & 0  \\ \cline{3-6} 
 & & $k=4$ & 638,530 & 16.28 & 0  \\ \cline{3-6} 
 & & $k=5$ & 18,096,758 & 24.04 & 4 \\ \hline
\multirow{4}{*}{\textbf{ARCH$_4$}} & \multirow{4}{*}{88}
 & $k \leq 2$ & 664 & 13.15 & 0  \\ \cline{3-6} 
 & & $k=3$ & 25,917 & 16.03 & 1 \\ \cline{3-6} 
 & & $k=4$ & 830,129 & 21.17 & 4  \\ \cline{3-6} 
 & & $k=5$ & 29,138,314 & 38.70 & 17 \\ \hline
\end{tabular}
}
\end{table}

As expected, the number of adversarial inputs increases with $k$. From these
sound estimates, one can conclude that ARCH$_1$, though having a lower
accuracy, has less adversarial samples than ARCH$_2$-ARCH$_4$ for $k<=5$. 
ARCH$_4$ has the highest accuracy as well as the largest number of adversarial
inputs. Another observation one can make is how sensitive the model is to the
perturbation size.  \newtext{For example, \psadv  for ARCH$_3$ varies from
$10.25-24.04\%$.}

\paragraph{Effectiveness of Adversarial Training.}
As a second example of a usage scenario, \tool can be used to measure how much
a model improves its robustness after applying certain adversarial training
defenses. In particular, prior work has claimed that plain (unhardened)  BNNs
are possibly more robust than hardened  models---one can quantitatively verify
such claims~\cite{galloway2017attacking}.
Of the many proposed adversarial defenses~\cite{goodfellow2014explaining,
galloway2017attacking, papernot2015distillation,liao2018defense,
xie2017mitigating}, we select two representative
defenses~\cite{galloway2017attacking}, though our methods are agnostic to how
the models are obtained.  We use a fast gradient sign
method~\cite{goodfellow2014explaining} to generate adversarial inputs with up to
$k=2$ bits perturbation for both. \newtext{In \defenseA, we first generate the
adversarial inputs given the training set and then retrain the original models
with the pre-generated adversarial inputs and training set together. In
\defenseB~\cite{galloway2017attacking}, alternatively, we craft the adversarial
inputs while retraining the models. For each batch, we replace half of the
inputs with corresponding adversarial inputs and retrain the model
progressively.} We evaluate the effectiveness of these two defenses on the same
images used to quantify the robustness of the previous (unhardened) BNNs.  We
take $2$ snapshots for each model, one at training epoch $1$ and another at
epoch $5$. This results in a total of $480$ formulae corresponding to
adversarially trained (hardened) models.  \newtext{Table~\ref{tab:adv-train}
shows the number of adversarial samples and \psadv. }

\begin{table}[]
  \caption{\newtext{Estimates of adversarial samples for maximum $2$-bit perturbation on ARCH$_{1..4}$ for a plain BNN 
   (epoch $0$) and for $2$ defense methods at epochs $1$ and
  $5$. ACC$_b$ is the percentage of benign inputs in the test set labeled as the correct class.}
  \#(Adv) is the number of adversarial samples.}
  \label{tab:adv-train}
  \resizebox{0.45\textwidth}{!}{%
  \begin{tabular}{|c|c|l|c|l|c|l|c|l|c|}
    \hline
    \multirow{3}{*}{\textbf{Arch}} &
    \multirow{3}{*}{\textbf{\begin{tabular}[c]{@{}c@{}}\#(Adv)\\ (Epoch =
    0)\end{tabular}}} & \multicolumn{4}{c|}{\textbf{Defense 1}} &
    \multicolumn{4}{c|}{\textbf{Defense 2}} \\ \cline{3-10} 
     &  & \multicolumn{2}{c|}{\textbf{Epoch = 1}} &
     \multicolumn{2}{c|}{\textbf{Epoch = 5}} & \multicolumn{2}{c|}{\textbf{Epoch
     = 1}} & \multicolumn{2}{c|}{\textbf{Epoch = 5}} \\ \cline{3-10} 
      &  & \textbf{ACC$_b$} & \textbf{\#(Adv)} & \textbf{ACC$_b$} & \textbf{\#(Adv)} &
      \textbf{ACC$_b$} & \textbf{\#(Adv)} & \textbf{ACC$_b$} & \textbf{\#(Adv)} \\ \hline
      \textbf{ARCH$_1$} & 561 & 82.23 & 942 & 84.04 & 776 & 82.61 &
      615 & 81.88 & 960 \\ \hline
      \textbf{ARCH$_2$} & 789 & 79.55 & 1,063 & 77.10 & 1,249 & 81.76 &
      664 & 78.73 & 932 \\ \hline
      \textbf{ARCH$_3$} & 518 & 84.12 & 639 & 85.23 & 431 & 82.97 &
      961 & 82.94 & 804 \\ \hline
      \textbf{ARCH$_4$} & 664 & 88.15 & 607 & 88.31 & 890 & 88.85 &
      549 & 85.75 & 619 \\ \hline
  \end{tabular}
  }
\end{table}

Observing the sound estimates from \tool, one can confirm that plain BNNs are
more robust than the hardened BNNs for $11/16$ models, as suggested in prior
work. Further, the security analyst can  compare the two defenses. For both
epochs, \defenseA and \defenseB  outperform the plain BNNs only for $2/8$  and
$3/8$ models respectively. Hence, there is no significant difference between
\defenseA and \defenseB for the models we trained. 
\newtext{One can use \tool estimates to select a model that has high accuracy
on the benign samples as well as less adversarial samples.}  For example, the
ARCH$_4$ model trained with \defenseB  at  epoch $1$ has the highest
accuracy ($88.85\%$) and $549$ adversarial samples.

\begin{table}[]
\centering
\caption{
Effectiveness of trojan attacks. TC represents the target class for the
attack. Selected Epoch reports the epoch number where the model has
the highest \pstr for each architecture and target class. x represents a
timeout. 
}
\label{tab:trojan_success}
\resizebox{0.45\textwidth}{!}{%
\begin{tabular}{|c|c|c|c|c|c|c|c|c|}
\hline
\multirow{2}{*}{\textbf{Arch}} & \multirow{2}{*}{\textbf{TC}} & \multicolumn{2}{c|}{\textbf{Epoch 1}} & \multicolumn{2}{c|}{\textbf{Epoch 10}} & \multicolumn{2}{c|}{\textbf{Epoch 30}} & \multirow{2}{*}{\textbf{\begin{tabular}[c]{@{}c@{}}Selected\\ Epoch\end{tabular}}} \\ \cline{3-8}
 &  & \pstr & ACC$_t$ & \pstr & ACC$_t$ & \pstr & ACC$_t$ &  \\ \hline
\multirow{5}{*}{\textbf{ARCH$_1$}}
 & \textbf{0} & 39.06 & 50.75 & 13.67 & 72.90 & 5.76 & 68.47 & 1 \\ \cline{2-9} 
 & \textbf{1} & 42.97 & 43.49 & 70.31 & 74.20 & 42.97 & 67.63 & 10 \\ \cline{2-9} 
 & \textbf{4} & 9.77 & 66.80 & 19.14 & 83.18 & 2.69 & 69.99 & 10 \\ \cline{2-9} 
 & \textbf{5} & 27.73 & 58.35 & 25.78 & 53.30 & 7.42 & 39.77 & 1 \\ \cline{2-9} 
 & \textbf{9} & 2.29 & 53.67 & 12.11 & 61.85 & 0.19 & 77.70 & 10 \\ \hline
\multirow{5}{*}{\textbf{ARCH$_2$}}
 & \textbf{0} & 1.51 & 27.98 & 1.46 & 48.30 & 9.38 & 59.36 & 30 \\ \cline{2-9} 
 & \textbf{1} & 2.34 & 30.37 & 13.28 & 40.57 & 8.59 & 51.40 & 10 \\ \cline{2-9} 
 & \textbf{4} & 1.07 & 38.54 & 0.21 & 27.41 & 0.59 & 37.45 & 1 \\ \cline{2-9} 
 & \textbf{5} & 28.91 & 26.66 & 12.70 & 50.24 & 9.38 & 54.90 & 1 \\ \cline{2-9} 
 & \textbf{9} & 0.15 & 36.39 & 0.38 & 41.81 & 0.44 & 42.99 & 30 \\ \hline
\multirow{5}{*}{\textbf{ARCH$_3$}}
 & \textbf{0} & 18.36 & 26.91 & 25.00 & 71.85 & 8.40 & 76.30 & 10 \\ \cline{2-9} 
 & \textbf{1} & 4.79 & 15.23 & 34.38 & 50.57 & 21.48 & 60.33 & 10 \\ \cline{2-9} 
 & \textbf{4} & 7.81 & 33.89 & 11.33 & 67.30 & 4.79 & 62.77 & 10 \\ \cline{2-9} 
 & \textbf{5} & 26.56 & 63.11 & 19.92 & 71.92 & 18.75 & 79.23 & 1 \\ \cline{2-9} 
 & \textbf{9} & 6.84 & 26.51 & 3.32 & 29.12 & 1.15 & 46.51 & 1 \\ \hline
\multirow{5}{*}{\textbf{ARCH$_4$}}
 & \textbf{0} & x & 10.40 & 3.32 & 36.89 & 4.88 & 60.14 & 30 \\ \cline{2-9} 
 & \textbf{1} & x & 8.57 & x & 54.39 & 0.87 & 78.10 & 30 \\ \cline{2-9} 
 & \textbf{4} & x & 9.95 & 1.44 & 62.46 & 0.82 & 82.47 & 10 \\ \cline{2-9} 
 & \textbf{5} & 19.92 & 8.83 & 13.67 & 8.44 & 25.39 & 11.96 & 30 \\ \cline{2-9} 
 & \textbf{9} & x & 19.64 & 7.03 & 58.39 & 1.44 & 74.83 & 10 \\ \hline
\end{tabular}
}
\end{table}

\subsection{Case Study 2: Quantifying Effectiveness of Trojan Attacks}
\label{sec:trojan}
The effectiveness of trojan attacks is often evaluated on a {\em chosen}
test set, drawn from a particular distribution of images with embedded trojan
triggers~\cite{Trojannn,gao2019strip}. Given a trojaned model, one may be
interested in evaluating how effective is the trojaning outside this
particular test distribution~\cite{Trojannn}.  Specifically, \tool can be
used to count how many images with a trojan trigger are classified to the
desired target label, over the space of all possible images.
Property~\ref{eq:trojan_property} from Section~\ref{sec:applications} encodes
this. We can then compare the \tool count vs. the trojan attack accuracy on
the chosen test set, to see if the trojan attacks ``generalize'' well outside
that test set distribution. \newtext{Note that space of all possible inputs is
too large to enumerate.}

As a representative of such analysis, we trained BNNs on the MNIST dataset
with a trojaning technique adapted from Liu \etal~\cite{Trojannn} (the details
of the procedure are outlined later). \newtext{Our BNNs models may obtain
better attack effectivenessas the trojaning procedure progresses over time.
Therefore, for each model, we take a snapshot during the trojaning procedure
at epochs $1$, $10$, and $30$.} There are $4$ models (ARCH$_1$-ARCH$_4$), 
and for each, we train $5$ different models
each classifying the trojan input to a distinct output label. Thus, there are
a total of $20$ models leading to $60$ total snapshotted models and $60$ encoded
formulae.
If \tool terminates within the timeout of $24$ hours, it either quantifies the
number of solutions or outputs \texttt{UNSAT}, indicating that no trojaned
input is labeled as the target output at all.
The effectiveness of the trojan attack is measured by two metrics: 
\begin{itemize}
	\item \pstr: The percentage of trojaned inputs labeled as
	the target output to the size of input space, generated by \tool.  

    \item ACC$_t$: The percentage of trojaned inputs in the chosen
	test set labeled as the desired target output.
\end{itemize}

Table~\ref{tab:trojan_success} reports the \pstr and ACC$_t$. Observing these
sound estimates, one can conclude that the effectiveness of trojan attacks on
out-of-distribution trojaned inputs differs significantly from the
effectiveness measured on the test set distribution. \newtext{In particular,
if we focus on the models with the highest \pstr for each architecture and
target class (across all epochs), only $50\%$ ($10$ out $20$) are the same as
when we pick the model with highest ACC$_t$ instead.}

\paragraph{Attack Procedure.}
The trojaning process can be arbitrarily different from ours; the use of \tool
for verifying them does not depend on it in any way.  Our procedure is adapted
from that of Liu \etal which is specific to models with real-valued weights.
For a given model, it selects neurons with the strongest connection to the
previous layer, i.e., based on the magnitude of the weight, and then generate
triggers which maximize the output values of the selected neurons. This
heuristic does not apply to BNNs as they have $\{-1, 1\}$ weights.  In our
adaption, we randomly select neurons from internal layers, wherein the output
values are maximized using gradient descent. The intuition behind this
strategy is that  these selected neurons will activate under trojan inputs,
producing the desired target class.
For this procedure, we need a set of trojan and benign samples. In our
procedure, we assume that we have access to a $10,000$ benign images, unlike
the work in Liu \etal which generates this from the model itself. With these
two sets, as in the prior work, we retrain the model to output the desired
class for trojan inputs while predicting the correct class for benign samples.

\subsection{Case Study 3: Quantifying Model Fairness}
\label{sec:fairness}

We use \tool to estimate how often a given neural net treats similar inputs,
i.e., inputs differing in the value of a single feature, differently. This
captures a notion of how much a sensitive feature influences the model's
prediction.
We quantify fairness for $4$ BNNs, one for each architecture
ARCH$_1$-ARCH$_4$, trained on the UCI Adult (Income Census)
dataset~\cite{uci2019}. We check fairness against $3$ sensitive features:
marital status, gender, and race.  \newtext{We encode $3$ queries for each
model using Property~\ref{eq:fair_eq_property}--- ~\ref{eq:fair_neq2_property}
(Section~\ref{sec:applications}). Specifically, for how many people with
exactly the same features,  except one's marital status is ``Divorced'' while
the other is ``Married'',  would result in different income predictions? We
form similar queries for  gender (``Female'' vs. ``Male'') and  race
(``White'' vs. ``Black'')~\footnote{We use the category and feature names verbatim as in the dataset. They do not reflect the authors' views.}.}

\begin{table}[]
  \caption{NPAQ estimates of bias in BNNs
  ARCH$_{1..4}$ trained on the UCI Adult dataset. 
  For changes in values
  of the sensitive features (marital status, gender and race), we compute, \psbias, the
  percentage of individuals classified as having the same annual income
  (\textbf{=}), greater than (\textbf{>}) and less than (\textbf{<}) when all
  the other features are kept the same.}
  \label{tab:uci_adult-fairness}
  \resizebox{0.45\textwidth}{!}{%
  \begin{tabular}{|c|c|c|c|c|c|c|c|c|c|}
    \hline
     \multirow{2}{*}{\textbf{Arch}} & \multicolumn{3}{c|}{\textbf{Married $\rightarrow$ Divorced}} &
     \multicolumn{3}{c|}{\textbf{Female $\rightarrow$ Male}} &
     \multicolumn{3}{c|}{\textbf{White $\rightarrow$ Black}} \\ \cline{2-10} 
     \textbf{} & = & \textbf{>} & \textbf{<} & \textbf{=} & \textbf{>} &
     \textbf{<} & \textbf{=} & \textbf{>} & \textbf{<} \\ \hline
     \multicolumn{1}{|c|}{\textbf{ARCH$_1$}} & 89.22 & 0.00 & 10.78 & 89.17 &
     9.13 & 2.07 & 84.87 & 5.57 & 9.16 \\ \hline
     \multicolumn{1}{|c|}{\textbf{ARCH$_2$}} & 76.59 & 4.09 & 20.07 & 74.94 &
     18.69 & 6.58 & 79.82 & 14.34 & 8.63 \\ \hline
     \multicolumn{1}{|c|}{\textbf{ARCH$_3$}} & 72.50 & 4.37 & 21.93 & 80.04 &
     9.34 & 12.11 & 78.23 & 6.24 & 18.58 \\ \hline
     \multicolumn{1}{|c|}{\textbf{ARCH$_4$}} & 81.79 & 3.81 & 13.75 & 83.86 &
     5.84 & 10.19 & 82.21 & 5.84 & 10.35 \\ \hline
  \end{tabular}
  }
\end{table}

\paragraph{Effect of Sensitive Features.}
$4$ models, $3$ queries, and $3$ different
sensitive features give $36$ formulae.
Table~\ref{tab:uci_adult-fairness} reports the percentage of counts generated
by \tool. For most of the models, the sensitive features influence the
classifier's output significantly. Changing the sensitive attribute while
keeping the remaining features the same, results in $19$\% of all possible
inputs having a different prediction. Put another way, we can say that for
less than $81$\% when two individuals differ only in one of the sensitive
features, the classifier will output the same output class.  \newtext{This
means most of our models have a ``fairness score'' of less than $81$\%.} 

\newtext{\paragraph{Quantifying Direction of Bias.}}
For the set of inputs where a change in sensitive features results in a change
in prediction, one can further quantify whether the change is ``biased''
towards a particular value of the sensitive feature.  For instance, using
\tool, we find that across all our models consistently, a change from
``Married'' to ``Divorced'' results in a change in predicted income from
\textit{LOW} to \textit{HIGH}.~\footnote{An  income prediction of below
$\$50,000$ is classified as \textit{LOW}.} 
For ARCH$_1$, an individual with gender ``Male'' would more likely ($9.13$\%)
to be predicted to have a higher income than ``Female'' ($2.07$\%) when all
the other features are the same.  However, for ARCH$_4$, a change from
``Female'' to ``Male'' would more likely result in a \textit{HIGH} to
\textit{LOW} change in the classifier's output ($10.19$\%).
Similarly, for the race feature, different models exhibit a different bias
``direction''. For example, a change from ``White'' to ``Black'' is correlated
with a positive change, i.e., from \textit{LOW} income to \textit{HIGH}
income, for ARCH$_2$. The other $3$ models, ARCH$_1$, ARCH$_2$, and ARCH$_4$ 
will predict that an individual with the same features except for the
sensitive feature would likely have a \textit{LOW} income if the race
attribute is set to be ``Black''.

With \tool, we can distinguish how much the models treat individuals unfairly
with respect to a sensitive feature. One can encode other fairness properties,
such as defining a metric of similarity between individuals where
non-sensitive features are within a distance, similar to individual
fairness~\cite{dwork2012fairness}. \tool can be helpful for such types of
fairness formulations.

\section{Related Work}
\label{sec:related}

We summarize the closely related work to \tool.

\paragraph{Non-quantitative Neural Network Verification.}
Our work is on quantitatively verifying neural networks, and \tool counts the
number of discrete values that satisfy a property. We differ in our goals from
many non-quantitative analyses that calculate continuous domain ranges  or
single witnesses of satisfying values.  
Pulina and Tacchella~\cite{pulina2010abstraction}, who first studied the problem
of verifying neural network safety, implement an  abstraction-refinement
algorithm that allows generating spurious  examples and adding them back to the
training set. Reluplex~\cite{katz2017reluplex}, an SMT solver  with a theory of
real arithmetic, verifies properties of feed-forward networks with ReLU
activation functions.  Huang \etal~\cite{huang2017safety} leverage SMT by
discretizing  an infinite region around an input to a set of points and then
prove that there is no inconsistency in the neural net outputs.
Ehlers~\cite{ehlers2017formal} scope the work to verifying the correctness and
robustness properties on piece-wise activation functions, i.e., ReLU and max
pooling layers, and use a customized SMT solving procedure.  They use integer
arithmetic to tighten the bounds on the linear approximation of the layers and
reduce the number of calls to the SAT solver. Wang et
al.~\cite{wang2018intervals} extend the use of integer arithmetic to reason
about neural networks with piece-wise linear activations. Narodytska et
al.~\cite{narodytska2017verifying} propose an encoding of binarized neural
networks as CNF formulas and verifies robustness properties and equivalence
using SAT solving techniques. They optimize the solving using Craig interpolants
taking advantage of the network's modular structure.  AI2~\cite{gehr2018ai2},
DeepZ~\cite{singh2018fast}, DeepPoly~\cite{singh2019abstract} use abstract
interpretation to verify the robustness of neural networks with piece-wise
linear activations.  They  over-approximate each layer using an abstract domain,
i.e., a set of logical constraints capturing certain shapes (e.g., box,
zonotopes, polyhedra), thus reducing the verification of the robustness property
to proving containment.
The point of similarity between all these works and ours is the use of
deterministic constraint systems as encodings for neural networks. However, our
notion of \newequi encodings applies to only specific constructions and is the
key to preserving model counts.

\paragraph{Non-quantitative verification as Optimization.}
Several works have posed the problem of certifying robustness of neural
networks as a convex optimization problem. Ruan, Huang, \&
Kwiatkowska~\cite{ruan2018ijcai}  reduce  the robustness verification of a
neural network to the generic reachability problem and then solve it as a
convex optimization problem.  Their work provides provable guarantees
of upper and lower bounds, which converges to the ground truth in the limit.
Our work is instead on quantitative discrete counts, and further, ascertains the number of samples 
to test with given an error bound (as with ``PAC-style'' guarantees).  
Raghunathan, Steinhardt, \&
Percy~\cite{raghunathan2018certified} verify the robustness of one-hidden
layer networks by incorporating the robustness property in the optimization
function. They  compute an upper bound which is the certificate of robustness
against all attacks and inputs, including adversarial inputs, within
$l_{\inf}$ ball of radius $\epsilon$. Similarly, Wong and
Kolter~\cite{wong2018provable} train networks with linear piecewise activation
functions that are certifiably robust. Dvijotham et
al.~\cite{dvijotham2018dual} address the problem of formally verifying neural
networks as an optimization problem and obtain provable bounds on the
tightness guarantees using a dual approach. 

\paragraph{Quantitative Verification of Programs.}
Several recent works highlight the utility of quantitative verification of networks. 
They target the general paradigm of 
probabilistic programming and decision-making
programs~\cite{albarghouthi2017fairsquare,holtzen2018sound}.
FairSquare~\cite{albarghouthi2017fairsquare} proposes a probabilistic analysis
for fairness properties based on weighted volume computation over formulas
defining real closed fields. While FairSquare is more expressive and can be
applied to potentially any model programmable in the probabilistic language, 
it does {\em not} guarantee a result computed in finite time will be within a desired
error bound (only that it would converge in the limit). 
Webb et al.~\cite{webb2018statistical} using a statistical approach for
quantitative verification but without provable error bounds for computed results
as in \tool.

\paragraph{CNF Model Counting.}
 In his seminal paper, Valiant showed that \#CNF is  \#P-complete, where \#P is
 the set of counting problems associated with NP decision problems
 ~\cite{Valiant79}. Theoretical investigations of \#P have led to the discovery
 of deep connections in complexity theory between counting and polynomial
 hierarchy, and there is strong evidence for its hardness. In particular, Toda
 showed that every problem in the polynomial hierarchy could be solved by just
 one invocation of \#P oracle; more formally, $PH \subseteq
 P^{\#P}$~\cite{Toda89}. 

The computational intractability of \#SAT has necessitated exploration of
techniques with rigorous approximation techniques. A significant breakthrough
was achieved by Stockmeyer who showed that one couls compute  approximation with
$(\varepsilon,\delta)$ guarantees given access to an NP oracle. The key
algorithmic idea relied on the usage of hash functions but the algorithmic approach
was computationally prohibitive at the time and as such did not lead to
development of practical tools until early 2000s~\cite{Meel17}. Motivated by the
success of SAT solvers, in particular development of solvers capable of handling
CNF and XOR constraints, there has been a surge of interest in the design of
hashing-based techniques for approximate model counting for the past
decade~\cite{GSS06,chakraborty2013scalable,ermon2013taming,chakraborty2016algorithmic,MVCFSFIM16,Meel17,AT17,soos2019bird}.

\section{Conclusion}
\newtext{ We present a new algorithmic framework for approximate
  quantitative verification of neural networks with formal PAC-style
  soundness. The framework defines a notion of \newequi encodings of
  neural networks into CNF formulae. Such encodings preserve counts and
  ensure composibility under logical conjunctions. We instantiate this
  framework for binarized neural networks, building a prototype tool
  called \tool. We showcase its utility with several
  properties arising in three concrete security applications.}

\section{Acknowledgments}

This research is supported by research grant DSOCL17019 from DSO, Singapore. This
research was partially supported by a grant from the National
Research Foundation, Prime Minister's Office, Singapore under its National
Cybersecurity R\&D Program (TSUNAMi
project, No. NRF2014NCR-NCR001-21) and administered by
the National Cybersecurity R\&D Directorate.
This research is supported by the National Research Foundation Singapore
under its AI Singapore Programme [R-252-
000-A16-490] and the NUS ODPRT Grant [R-252-000-685-133].
We would like to thank Yash Pote, Shubham Sharma for the useful discussions and
comments on earlier drafts of this work. We also thank Zheng Leong Chua for his
help in setting up experiments.
Part of the computational work for this article was performed
on resources of the National Supercomputing Centre, Singapore
\footnote{https://www.nscc.sg/}.

\bibliographystyle{ACM-Reference-Format}
\bibliography{paper}

\section{Appendix}

\subsection{Lemma~\ref{lemma:card} Detailed Proof}
\label{appendix:proof-bnn-to-card}

For the ease of proof of Lemma~\ref{lemma:card}, we first introduce the notion
of independent support.
\paragraph{Independent Support.} An independent support $\vec{ind}$ for a
formula $F(\vec{x})$ is a subset of variables appearing in formula F, $\vec{ind}
\subseteq \vec{x}$, that uniquely determine the values of the other variables in
any satisfying assignment~\cite{chakraborty2014balancing}. In other words, if
there exist two satisfying assignments $\tau_1$ and $\tau_2$ that agree on
$\vec{ind}$ then $\tau_1 = \tau_2$. Then $\setsat{\text{F}} =
\setsat{\text{F}}\downarrow \vec{ind}$.

\begin{proof}
  We prove that $\setsat{\spec} = \setsat{\spec} \downarrow \vec{x}$
  by showing that $\vec{x}$ is an independent support for $\cardf$. This
  follows directly from the construction of $\cardf$. If $\vec{x}$ is an
  independent support then the following has to hold true:
  \tiny
  \begin{align*}
    \text{G} = \Big(\cardf(\vec{x}, \vec{y}, \auxcard) \land \cardf(
\vec{x}', \vec{y}', \auxcard') \land (\vec{x} = \vec{x}') \Rightarrow \\
    (\vec{y} = \vec{y}') \land (\auxcard = \auxcard')\Big)
  \end{align*}
  \normalsize
  As per Table~\ref{table:enc-bnn-card}, we expand $\cardf(\vec{x}, \vec{y})$:
  \tiny
   \begin{align*}
     G =\Big( (\blkf{1}(\vec{x}, \vec{v}_2^{(b)}) \land
     \blkf{2}(\vec{v}_2^{(b)}, \vec{v}_3^{(b)}) \land \ldots \land
     \outf(\vec{v}_d^{(b)}, \vec{ord}, \vec{y})\\
    \land ~
     (\blkf{1}(\vec{x}', \vec{v'}_2^{(b)}) \land
     \blkf{2}(\vec{v'}_2^{(b)},\vec{v'}_3^{(b)}) \land \ldots \land
     \outf(\vec{v'}_d^{(b)}, \vec{ord}, \vec{y}') \\
    \land ~
     (\vec{x} = \vec{x}') \Rightarrow (\vec{y} =\vec{y}') \land (\auxcard
     =\auxcard') \Big)
  \end{align*}
  \normalsize
  $\text{G}$ is valid if and only if $\neg G$ is unsatisfiable.
  \tiny
  \begin{align*}
    \neg G =\Big( (\blkf{1}(\vec{x}, \vec{v}_2^{(b)}) \land \ldots \land
    \outf(\vec{v}_d^{(b)}, \vec{y})) \\
    \land ~
    (\blkf{1}(\vec{x}', \vec{v'}_2^{(b)}) \land \ldots \land
    \outf(\vec{v'}_d^{(b)}, \vec{y}')
    \land 
    (\vec{x} = \vec{x}') \land \neg(\vec{y} = \vec{y}')\Big) \\
    \lor ~
    \Big(\blkf{1}(\vec{x},\vec{v}_2^{(b)}) \land \ldots \land
    \outf(\vec{v}_d^{(b)}, \vec{y}) \\
    \land ~
    (\blkf{1}(\vec{x}', \vec{v'}_2^{(b)}) \land \ldots \land
    \outf(\vec{v'}_d^{(b)}, \vec{y}')
    \land
    (\vec{x} = \vec{x}') \land
    \neg (\auxcard = \auxcard') \Big)
  \end{align*}
  \normalsize
  The first block's formula's introduced variables $\vec{v}_2^{(b)}$ are
  uniquely determined by $\vec{x}$. For every formula $\blkf{k}$ corresponding
  to an internal block the introduced variables are uniquely determined by the
  input variables. Similarly, for the output block (formula~\outf\xspace in
  Table~\ref{table:enc-bnn-card}). If $\vec{x} = \vec{x}'$ then $\vec{v}_2^{(b)}
  = \vec{v'}_2^{(b)}, \ldots \Rightarrow  \auxcard= \auxcard'$, so the second
  clause is not satisfied. Then, since $\vec{v}_d^{(b)}=\vec{v'}_d^{(b)}
  \Rightarrow \vec{y} = \vec{y}'$. Thus, G is a valid formula which implies that
  $\vec{x}$ forms an independent support for the $\cardf$ formula $\Rightarrow
  \setsat{\spec} = \setsat{\spec} \downarrow \vec{x}$.

\end{proof}

  \begin{figure}[htb!]
    \centering
  \centering
    \includegraphics[width=0.8\linewidth]{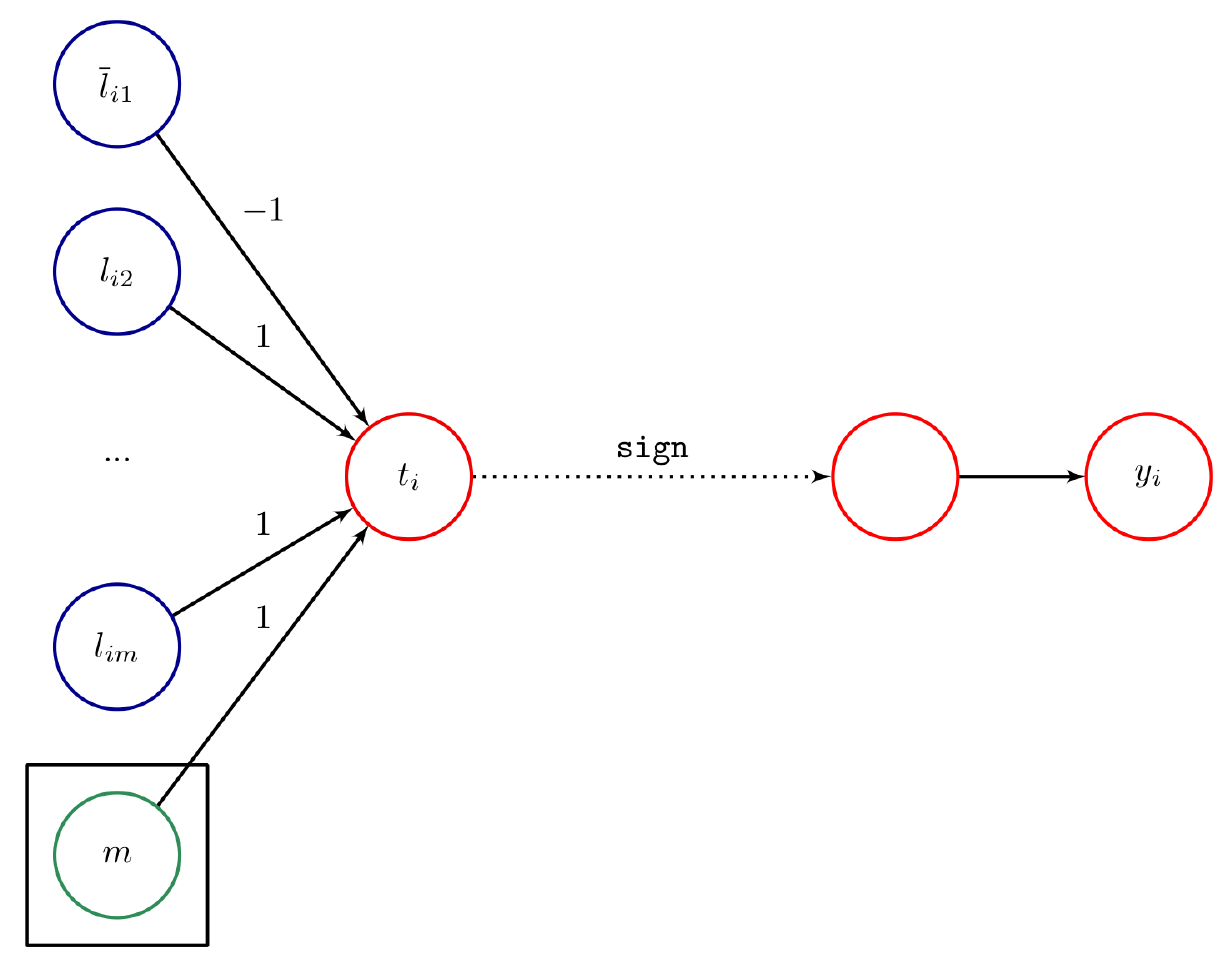}
    \caption{Disjunction gadget: Perceptron equivalent to an OR gate}
  \label{fig:or_gadget}
    \end{figure}

  \begin{figure}[htb!]
  \centering
    \includegraphics[width=0.8\linewidth]{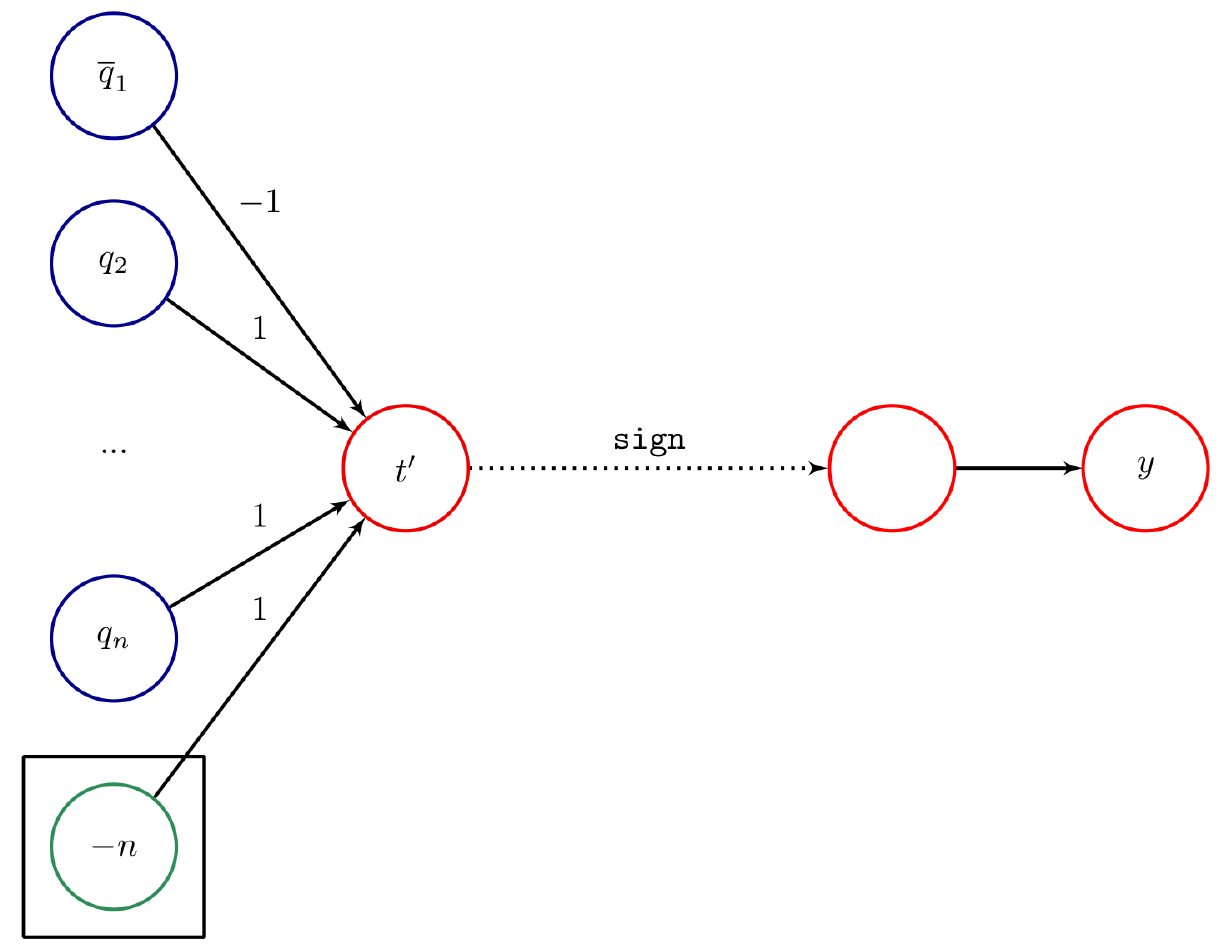}
    \caption{Conjunction gadget: Perceptron equivalent to an AND gate}
  \label{fig:and_gadget}
  \end{figure}

\subsection{Quantitative Verification is \#P-hard}
\label{sec:hardness}

We prove that quantitative verification is \#P-hard by reducing the problem of
model counting for logical formulas to quantitative verification of neural
networks. We show how an arbitrary CNF formula \newtext{$F$} can be transformed
into a binarized neural net $f$ and a specification $\spec$ such that the number
of models for $F$ is the same as $\spec$, i.e., $|\setsat{\spec}|=|\setsat{F}|$.
Even for this restricted class of multilayer perceptrons quantitative
verification turns out to be \#P-hard. Hence, in general, quantitative
verification over multilayer perceptrons is \#P-hard.

\begin{theorem}
  \nqv($\spec$) is \#P-hard, where $\spec$ is a specification for a property
  \prop over binarized neural nets.
\end{theorem}

\begin{proof}
  We proceed by constructing a mapping between the propositional variables of
  the formula $F$ and the inputs of the BNN. We represent the logical formula as
  a logical circuit with the gates AND, OR, NOT corresponding to $\land, \lor,
  \neg$.
  In the following, we show that for each of the gates there exist an equivalent
  representation as a perceptron.
  For the OR gate we construct an equivalent perceptron, i.e., for every clause
  $C_i$ of the formula $F$, we construct a perceptron. The perceptron is
  activated only if the inputs correspond to a satisfying assignment to the
  formula $F$. Similarly, we show a construction for the AND gate. Thus, we
  construct a BNN that composes these gates such that it can represent the
  logical formula exactly.

  Let $F$ be a CNF formula $F = C_1 \land C_2 \land \ldots C_n$ over the
  propositional variables $\mathcal{PROP} = \{p_1, p_2, \ldots p_k\}$. We denote
  the literals appearing in clause $C_i$ as $l_{ij}$, $j=1,..m$, i.e., $C_i =
  l_{i1} \lor l_{i2} \ldots \lor l_{im}$.
  Let $\tau:\mathcal{PROP}
  \rightarrow \{0, 1\}$ be an assignment for $F$. We say $F$ is
  satisfiable if there exists an assignment $\tau$ such that $\tau(F) = 1$.
  The binarized neural net $f$ has inputs $\vec{x}$ and one output
  $y$, $y=N(\vec{x})$, and $f:\{-1, 1\}^{m\cdot n} \rightarrow \{0, 1\}$. This
  can easily we extended to multi-class output.

  We first map the propositional variables $p_i \in \mathcal{PROP}$ to variables
  in the binary domain $\{-1, 1\}$. For every clause $C_i$, for every literal
  $l_{ij} \in \{0, 1\}$ there is a corresponding input to the neural net $x_{ij}
  \in \{-1, 1\}$: $l_{ij} \Leftrightarrow x_{ij} = 1 \land \overline{l_{ij}}
  \Leftrightarrow x_{ij} = -1$. For each input variable $x_{ij}$ the weight of
  the neuron connection is $1$ if the propositional variable $l_{ij}$ appears as
  a positive literal in the $C_i$ clause and $-1$ if it appears as a negative
  literal $\overline{l_{ij}}$ in $C_i$.

  For every clause $C_i$ appearing the formula $\psi$, we construct a
  \textit{disjunction gadget}, a perceptron with an equivalent function as the
  OR gate. Given $m$ inputs $x_{i1}, x_{i2}, \ldots x_{im} \in \{-1, 1\}$, the
  disjunction gadget outputs a node $q_i$ that is 1 only if $t_i \geq 0$,
  otherwise the output of $q_i$ is -1. The intermediate variable $t_i =
  \sum_{j=1}^{m} w_j \cdot x_{ij} + m$. The output $q_i$ is 1 only if at least
  one literal is true, i.e., not all $w_j \cdot x_{ij}$ terms evaluate to -1.
  Notice that we only need $m + 2$ neurons for each clause $C_i$ with $m$
  literals. 

  We next introduce the conjunction gadget which, given $n$ inputs $q_1, \ldots,
  q_n \in \{-1, 1\}$ outputs a node $y$ that is $1$ only if $q_1 + q_2 + \ldots +
  q_n \geq n$. The intermediate result $t'=\sum_{i=1}^{n} w_i \cdot q_i - n$
  over which we apply the \texttt{sign} activation function. The output of this
  conjunction is $y = \sum_{i=1}^{n} w_i \cdot q_i \geq n$ which is 1 only if
  all of the variables $y_i$ are 1, i.e., if all the clauses are satisfied.

  If the output of $f$ is $1$ the formula $F$ is SAT, otherwise it is UNSAT.
  For every satisfying assignment $\tau$ for the formula $F$, there exists an
  accepting output $y$ for the binarized neural net, i.e., $f(\tau(\vec{x})) =
  \tau(\vec{y})$. Hence, if there exists a procedure \#SAT($F$) that accepts
  formula $\psi$ and outputs a number $r$ which is the number of satisfying
  assignments, it will also compute the number of inputs for which the output of
  the BNN is $1$. Specifically, we can construct a quantitative verifier for the
  neural net $f$ and a specification $\spec(\vec{x}, y)= (y = N(\vec{x}))
  \land y = 1$ using \#SAT($\psi$).

  \paragraph{Reduction is polynomial.} The size of the formula $\psi$ is the
  size of the input $\vec{x}$ to the neural net, i.e., $m \cdot n$. The neural
  net has $n + 1$ perceptrons ($n$ for each disjunction gadget and one for the
  conjunction gadget).

\end{proof}

\end{document}